\documentclass[12pt]{article}

\usepackage{lmodern}
\usepackage{geometry}
\usepackage{fullpage}
\usepackage{microtype}
\usepackage[utf8]{inputenc}
\usepackage[T1]{fontenc}
\usepackage{setspace}

\usepackage[round]{natbib}

\usepackage{hyperref}
\usepackage[title]{appendix}

\usepackage{textcomp}

\usepackage{graphicx}
\usepackage{floatrow}

\usepackage{caption, subcaption}

\usepackage{booktabs,threeparttable,array}

\usepackage{amsmath, amssymb, amsthm}
\theoremstyle{plain}
\newtheorem{theorem}{Theorem}
\newtheorem{lemma}{Lemma}
\newtheorem{proposition}[theorem]{Proposition}
\newtheorem{corollary}{Corollary}
\theoremstyle{definition}

\newtheorem{example}{Example}
\theoremstyle{remark}

\newcommand{\E}{\mathsf{E}}

\newcommand{\var}{\mathrm{var}}
\newcommand{\cov}{\mathrm{cov}}

\newcommand{\pr}{\mathrm{P}}
\newcommand{\dd}{\mathrm{d}}
\DeclareMathOperator{\vecop}{vec}

\newcommand{\tr}{\operatorname{tr}}

\newcommand{\bdiag}{\operatorname{bdiag}}
\newcommand{\tsp}{\mathsf{T}}
\newcommand{\rN}{\mathrm{N}}

\newcommand{\R}[1]{\mathbb{R}^{#1}}

\DeclareMathOperator*{\argmax}{arg\,max}

\newcommand{\webappendix}{Appendix}

\begin{document}

\title{Score-based confidence intervals for variance-covariance parameters in linear mixed models}

\author{Matias Shedden \quad Karl Oskar Ekvall \\
{\normalsize Department of Statistics, University of Florida} \\
{\tt \normalsize sheddenmatias@ufl.edu \quad k.ekvall@ufl.edu }}

\maketitle

\onehalfspacing

\abstract{ We propose confidence intervals with near-nominal coverage of
variance-covariance parameters in linear mixed models, even when interest or
nuisance parameters are near the boundary, that is, when some random effects
have small variances or are strongly correlated. In such settings, our
simulations show standard Wald, likelihood ratio, and score intervals with
nominal level 95\% can have coverage as low as 45\% and as high as 99\%, while
the proposed intervals are near nominal. The proposed intervals
invert a modified profile score statistic based on the restricted likelihood,
with nuisance parameters estimated on an extended parameter set. In a
longitudinal study of socialization in children with autism, where the estimated
correlation between the random intercept and slope is near minus one, our
intervals show that children who start higher improve more slowly. In a genetic study of white blood cell count in mice, they
distinguish the contributions of variant classes even though several intervals
include zero. Compared to universal inference intervals, which are valid in finite
samples, ours are 1.7 to 3.1 times narrower, and finite where one universal
interval is not, at a forty-seventh of the computing time. We establish asymptotically correct coverage uniformly in the
parameter under a condition known to hold with independent clusters and crossed
random effects. }

\vspace{1em}
\vspace{1em}

\section{Introduction} \label{sec:intro}

Many applications of linear mixed models require inference on
(variance-)covariance parameters near the boundary of the parameter set, that
is, a variance near zero or a correlation near $\pm 1$. In longitudinal studies,
for instance, the variance of a random slope in time quantifies how much the
rate of change varies between individuals, and its covariance with the random
intercept indicates whether starting higher is associated with faster or slower
change. Section~\ref{subsec:autism} illustrates this with socialization scores
of children with autism at five ages, where the estimated correlation between
the random intercept and slope is near $-1$. In genetics, a variance component
can model the variability in a trait attributable to a class of genetic
variants, and the questions are how much each class contributes and how the
classes rank \citep{yang2010,yang2015}. In Section~\ref{subsec:snp}, the trait
is white blood cell count in mice and the classes are single-nucleotide
polymorphisms grouped by minor allele frequency, and several classes have small
or zero estimated variance components. Near the boundary, standard confidence
intervals for the covariance parameters are unreliable, and existing valid
ones are either restricted to special cases or can be very wide.

The asymptotic approximations underlying Wald and likelihood
ratio intervals are typically inaccurate for parameters on or near the boundary
\citep{stern2000likelihood}, sometimes
severely so. Our simulations include settings in which nominal $95\%$ likelihood
ratio intervals have coverage as low as $0.65$ and Wald intervals as low as
$0.75$, even though the true parameter is interior. In other settings, the Wald
interval has coverage up to $0.995$ and the likelihood ratio interval up to
$0.992$. Which parameter values are close enough to the boundary to affect
coverage depends on the model and the interval, but since the true parameter is
unknown, any method whose coverage depends substantially on its value is unreliable. We
provide computationally efficient, score-based confidence intervals and show
that they have asymptotically correct coverage uniformly in the parameter. In the 
simulations, they have near-nominal coverage in every setting.

Testing whether the parameter is on the boundary is well understood in general
\citep{self1987asymptotic,geyer1994asymptotics,silvapulle1995score,andrews2001testing}
and for mixed models in particular
\citep{stram1994variance,greven2008restricted,baey2019asymptotic}.
However, such a test requires an accurate distribution approximation only at the
boundary. A confidence interval requires accurate approximations also at
interior values arbitrarily close to the boundary. Existing resampling methods
do not resolve the difficulty: the bootstrap can be inconsistent when the true
parameter is on the boundary \citep{andrews2000bootstrap}, and modifications of
it for testing boundary points in mixed models do not extend straightforwardly
to confidence intervals \citep{sinha2009bootstrap,guedon2024bootstrap}, nor do
permutation tests for random effects
\citep{fitzmaurice2007note}.

Recent theory shows that the score statistic standardized by expected
information often remains approximately chi-squared near the boundary
\citep{ekvall2022confidence,ekvall2026uniform}. \citet{battey2024anomaly}
similarly find that statistic to be unaffected by an atypical geometry of the
likelihood that makes the Wald statistic ineffectual for a variance component.
However, existing reliable methods apply either in special cases, notably a
single variance component or a ratio of two
\citep{harville1985confidence,crainiceanu2004likelihood,schweiger2016fast,zhang2025fast},
or to the vector of all covariance parameters jointly, using confidence regions
\citep{ekvall2026uniform}. Such regions are difficult to interpret, and
scientific questions are often about a particular component. Projecting the
region onto one axis is often too conservative, and requires computing the full
region. Universal inference gives valid intervals under minimal conditions
\citep{zhang2026universal}, but, as the examples of
Section~\ref{sec:data_examples} show, they are much wider and slower to compute
than ours, which invert a profile score statistic.

Intervals for one component at a time have the others as nuisance parameters,
whose estimation, for example by profiling the likelihood, can be
computationally intensive and affect coverage. A hard case is a nuisance
parameter near the boundary, as in the genetics example of
Section~\ref{subsec:snp}, where one variance component is estimated to be zero
and several are nearly so. We address both computing and coverage by
extending the parameter set to include nuisance parameter values that are
incompatible with the random-effects formulation but still yield a valid
marginal likelihood. The intervals remain interpretable because the interest
parameters are constrained to their original parameter set. Extended parameter
sets have precedents in and outside mixed models
\citep{feng1992statistical,verbeke2003use,ketz2018subvector,stringer2026testing},
and nuisance variances on the boundary have been treated in testing
\citep{qu2013linear,wood2013simple}, but none of these gives confidence
intervals with a coverage guarantee in this setting.

Proving that the intervals attain their nominal coverage uniformly requires
asymptotics along sequences of true parameters approaching the boundary, whereas
classical theory for mixed models assumes a fixed and interior parameter
\citep{miller1977mle,jiang1996reml}. Building on recent theory for
general parametric models \citep{shedden2026profile}, we establish the
asymptotic distribution of the profile score statistic under such sequences,
with the nuisance parameters estimated on the extended set. Although the general
conditions are considerably more involved than what suffices without nuisance
parameters, we show that for the linear mixed model they reduce, with the
right choice of nuisance estimator, to one already known to hold in models with
independent clusters and crossed random effects \citep{ekvall2026uniform}.

Inverting the score test requires many evaluations of the expected information,
which makes naive implementations slow. Our implementation in the R package
\texttt{reconf} chooses automatically between a sparse and a dense way of
computing it, from the number of random effects and the sparsity of their design
matrix. The first works with matrices of the dimension of the random effects and
exploits their sparsity, as in longitudinal and crossed designs. The
second suits the dense relationship matrices of genetics, working with matrices
of the dimension of the data at a cost that does not grow with the number of
random effects. Section~\ref{sec:computation} describes both, together with the
warm-started search for the interval endpoints. In the longitudinal example of
Section~\ref{subsec:autism}, \texttt{reconf} returns finite intervals for all
four covariance parameters in half a second, roughly an order of magnitude
faster than the profile intervals of \texttt{lme4} \citep{bates2015fitting}.

% model.tex

\section{Model and proposed intervals} \label{sec:model}

\subsection{Restricted likelihood}
For a response vector $Y \in \R{n}$ and non-stochastic $X \in \R{n\times p}$ and
$Z \in \R{n\times q}$, suppose
\begin{align} \label{eq:lmm}
  Y = X\beta + Z U + E,
\end{align}
where the random effects vector $U \sim \rN_{q}(0, \Psi)$ and the error vector
$E \sim \rN_n(0, \psi_r I_n)$ are independent. We assume $\Psi =
\Psi(\psi) = \sum_{j = 1}^{r - 1} \psi_j H_j$ for known symmetric $H_j \in
\R{q\times q}$, so $\Psi$ is parameterized by the first $r - 1$ elements of $\psi =
[\psi_1, \dots, \psi_r]^\tsp$. The model in \eqref{eq:lmm} implies
\begin{align} \label{eq:marginal}
  Y \sim \rN_n(X\beta, \Sigma),~~ \Sigma = \Sigma(\psi) = Z \Psi(\psi) Z^\tsp + \psi_r I_n.
\end{align}
Variance components models are a special case with $Z = I_n$ and $H_j$
positive semi-definite; then $\Sigma = \sum_{j = 1}^{r -
1}\psi_j H_j + \psi_r I_n$ with $\psi_j \geq 0$, $j \in \{1, \dots, r\}$.

We focus on inference on elements of $\psi$ using the restricted likelihood
\citep{patterson1971recovery}. The restricted likelihood is common in practice and tends to give confidence intervals with coverage closer to the nominal level than those from the full likelihood. Inference based on the profile likelihood obtained by partially maximizing the log-likelihood of
\eqref{eq:marginal} over $\beta$ can be treated similarly; it is implemented
in the software, and \webappendix~\ref{sec:lmm_computing} gives details. 

To define the restricted likelihood, assume $X$ has full column rank and let
$V \in \R{n \times (n - p)}$ be semi-orthogonal with $V^\tsp X = 0$. Then $V^\tsp Y \sim \rN_{n - p}(0, \tilde{\Sigma})$,
$\tilde{\Sigma} = \tilde{\Sigma}(\psi) = V^\tsp \Sigma(\psi) V$, which does not
depend on $\beta$. The restricted log-likelihood is, up to a constant,
\begin{align} \label{eq:res_ll}
  \ell(\psi) = -\frac{1}{2}\{\log \vert \tilde{\Sigma}(\psi)\vert +
  Y^\tsp V \tilde{\Sigma}(\psi)^{-1} V^\tsp Y\},
\end{align}
where $\vert \cdot \vert$ denotes the determinant. 
Any two admissible $V$ are related by right-multiplication with an orthogonal matrix, so \eqref{eq:res_ll} does not depend on the choice of $V$. 
Up to an additive constant, \eqref{eq:res_ll} equals
\begin{align} \label{eq:res_ll2}
   -\frac{1}{2}[\log \vert \Sigma(\psi)\vert + \log \vert X^\tsp
  \Sigma(\psi)^{-1} X\vert + \{Y - X\tilde{\beta}(\psi)\}^\tsp
  \Sigma(\psi)^{-1} \{Y - X\tilde{\beta}(\psi)\}],
\end{align}
where $\tilde{\beta}(\psi) = \{X^\tsp \Sigma(\psi)^{-1}X\}^{-1} X^\tsp
\Sigma(\psi)^{-1} Y$. Computing is based on \eqref{eq:res_ll2}, while theory is often more convenient with \eqref{eq:res_ll}.

\subsection{Test statistic and confidence intervals}
Upon reordering if needed, we may assume the parameter of interest is $\psi_1$
and write $\psi = [\psi_1, \psi_{(2)}^\tsp]^\tsp$, where $\psi_{(2)} = [\psi_2,
\dots, \psi_r]^\tsp$ is a nuisance parameter. By \eqref{eq:lmm}, $\psi_1$ is
either a variance or a covariance. Conformably partition
the score $S(\psi)  = \nabla \ell(\psi) = [S_1(\psi), S_{(2)}(\psi)^\tsp]^\tsp$.
Similarly, the expected information $\mathcal{I}(\psi)= \cov \{S(\psi)\} \in \R{r\times r}$ has
blocks $\mathcal{I}_{(11)}\in \R{}$, $\mathcal{I}_{(12)} \in \R{1\times (r -
1)}$, and $\mathcal{I}_{(22)} \in \R{(r - 1)\times (r - 1)}$. Blocks of the
inverse are denoted with superscripts; in particular, $\mathcal{I}^{(11)} =
(\mathcal{I}_{(11)} -
\mathcal{I}_{(12)}\mathcal{I}_{(22)}^{-1}\mathcal{I}_{(21)})^{-1}$ is the
inverse (reciprocal) of the efficient information for $\psi_1$.

For a hypothesized value of $\psi_1$, let $\tilde{\psi}_{(2)} = \tilde{\psi}_{(2)}(\psi_1)$ be a nuisance parameter estimator to be discussed shortly; the choice critically
affects the properties of the resulting test statistic. Let also $\tilde{\psi} = [\psi_1, \tilde{\psi}_{(2)}^\tsp]^\tsp$. The proposed test statistic is the score for the interest parameter standardized
by the efficient information:
\begin{equation} \label{eq:standardized_score}
  T(\psi_1) = \mathcal{I}^{(11)}(\tilde{\psi})\, \{S_1(\tilde{\psi})\}^2.
\end{equation}
The expected information cannot be replaced by the observed information as the
coverage guarantee of Section~\ref{sec:gen_theory} can then fail
\citep[see][]{ekvall2022confidence}, and the observed information can be
negative in practice. The efficient information is needed because, when the nuisance parameters are
unknown, the information about $\psi_1$ is smaller than $\mathcal{I}_{(11)}$. The loss of information can be quantified using 
\begin{equation} \label{eq:kappa}
  \kappa(\psi)^2 = \frac{\mathcal{I}_{(12)}\mathcal{I}_{(22)}^{-1}\mathcal{I}_{(21)}}{\mathcal{I}_{(11)}},
\end{equation}
with all quantities evaluated at $\psi$. Note that $(\mathcal{I}^{(11)})^{-1} = \mathcal{I}_{(11)}\{1 -
\kappa(\psi)^2\}$ and, when $\psi$ is the true parameter, $\kappa(\psi)$ is the
multiple correlation between $S_1(\psi)$ and $S_{(2)}(\psi)$.

Inverting the test based on $T$ in \eqref{eq:standardized_score} gives the
$1 - \alpha$ confidence set
\[
  C(\alpha) = \left\{ \psi_1 : T(\psi_1) \leq q_{1-\alpha} \right\},
\]
where $q_{1-\alpha}$ is the $(1-\alpha)$-quantile of the
$\chi^2_{1}$-distribution and $\psi_1$ ranges over $[0, \infty)$ if it is a variance and over $\R{}$ if it is a covariance.

\subsection{Nuisance parameter estimators}\label{sec:nuis}

Assuming the parameter of interest is not the error variance, let $\Upsilon =
\{\psi : \Psi(\psi) \succeq 0,~ \psi_r > 0\}$ denote the parameter set, where
$\succeq$ denotes positive semi-definiteness. When the error variance is the
parameter of interest, the same arguments apply with a change of notation. Consider a
hypothesized $\psi_1$ and define the slice $\Upsilon^{(2)}(\psi_1) =
\{\psi_{(2)} : [\psi_1, \psi_{(2)}^\tsp]^\tsp \in \Upsilon\}$. That is, for
a given $\psi_1$, $\Upsilon^{(2)}(\psi_1)$ is the set of nuisance parameters
for which $\Psi$ is positive semi-definite and the error variance is positive. Profiling the restricted likelihood over $\psi_{(2)}$ in the usual way corresponds to taking $\tilde{\psi}_{(2)}(\psi_1) \in \argmax_{\psi_{(2)} \in \Upsilon^{(2)}(\psi_1)} \ell(\psi_1, \psi_{(2)})$. However, such estimators are often on the boundary of $\Upsilon^{(2)}(\psi_1)$, which typically leads to confidence intervals with far from nominal coverage; see Section \ref{sec:simulations}.

The proposed confidence intervals instead use a nuisance estimator over a
larger set. Let $\tilde{\Upsilon} = \{\psi : \tilde{\Sigma}(\psi) \succ 0\}$,
where $\succ$ denotes positive definiteness, and define the extended set
$\tilde{\Upsilon}^{(2)}(\psi_1) = \{\psi_{(2)} : [\psi_1, \psi_{(2)}^\tsp]^\tsp
\in \tilde{\Upsilon}\}$. Over this set, the restricted likelihood is still well-defined,
even though the interpretation offered by \eqref{eq:lmm} is not available in
general.
In particular, there are $\psi_{(2)} \in \tilde{\Upsilon}^{(2)}(\psi_1)$ such
that $\Psi(\psi_1, \psi_{(2)})$ is not a covariance matrix. Nevertheless, the
resulting methods lead to confidence intervals with the desired coverage
properties, as shown in Section~\ref{sec:gen_theory}. The extended set
$\tilde{\Upsilon}^{(2)}(\psi_1)$ is open, so any maximizer solves
$S_{(2)}(\psi_1, \tilde{\psi}_{(2)}) = 0$ and can in principle be found by
unconstrained optimization.

In large-scale applications, computing the exact maximizer for many candidate $\psi_1$ can be too
expensive. For such settings we use a
one-step approximation. Let $\mathring{\psi}_{(2)}$ be a preliminary estimator,
such as the restricted maximum likelihood estimator obtained when fitting the
model, which does not depend on $\psi_1$. Then, for each $\psi_1$, let
\begin{equation} \label{eq:onestep}
  \check{\psi}_{(2)}(\psi_1) = \mathring{\psi}_{(2)} +
    \mathcal{I}_{(22)}(\psi_1, \mathring{\psi}_{(2)})^{-1}
    S_{(2)}(\psi_1, \mathring{\psi}_{(2)})
\end{equation}
be the one-step update from it. Theoretical support for both choices, including
guarantees that they are well-defined for large enough sample sizes, is given
in Section~\ref{sec:gen_theory}.

\subsection{Computation} \label{sec:computation}

Making the method practical requires evaluating the test statistic $T(\psi_1)$
efficiently at the many values needed to invert the test. 
Each evaluation
requires solving an optimization problem to find the nuisance parameter estimator
$\tilde{\psi}_{(2)}(\psi_1)$. 
This is a nonconvex problem whose objective can be
irregular. 
Because $\tilde{\Upsilon}^{(2)}(\psi_1)$ is open, the problem is
unconstrained; over $\Upsilon^{(2)}(\psi_1)$ it would have positive
semi-definiteness constraints. We use a trust-region algorithm with the expected information, which
the statistic already requires, in place of the Hessian \citep[cf. the average-information
algorithm of][]{gilmour1995average}. When the random-effects design is
sparse, each evaluation uses a Woodbury identity for $\Sigma(\psi)^{-1}$ and
that sparsity, so the value, score, and expected information together cost
$O(nq + rq^3 + r^2 q^2)$ in the worst case and less when the factorizations
stay sparse. When $q \geq n$ and $Z$ is dense, as in kernel or
kinship models, that cost is dominated by dense $q \times q$ algebra. Our
implementation then forms the matrices $Z H_j Z^\tsp$ once and works with
$n \times n$ matrices thereafter, so each later evaluation costs
$O(rn^3 + r^2n^2)$, for any $q$.

To invert the test, we search outward from the estimate until the value of the
test-statistic crosses the critical value. The nuisance parameter estimate at each accepted value warm-starts the next
optimization, which speeds up the search. Step lengths
follow from a secant prediction of the critical-value crossing, refined by
regula falsi, and a limit typically needs 15--20 evaluations of the statistic.
The R package \texttt{reconf} implements these methods. \webappendix~\ref{sec:lmm_computing} gives additional computational details.

% asymptotics.tex

\section{Asymptotic coverage} \label{sec:gen_theory}

We next establish conditions under which the profile intervals $C_n(\alpha)$
have asymptotically correct uniform coverage probability on appropriate subsets of the parameter set $\Upsilon$. Specifically, for $A \subseteq \Upsilon$ to be specified,
\begin{equation} \label{eq:uniform_coverage}
  \lim_{n\to \infty}\sup_{\psi \in A}\left \vert \pr_\psi\{\psi_1 \in C_n(\alpha)\}
  - (1-\alpha)\right\vert = 0.
\end{equation}

We focus on sets $A$ containing sequences of parameters tending to the boundary
of $\Upsilon$. Such an $A$ may contain boundary parameters, or only interior parameters
arbitrarily close to the boundary. Establishing \eqref{eq:uniform_coverage} for
such sets is what makes coverage reliable however close the true parameter
is to the boundary. In some examples $A$ can be the full parameter set
$\Upsilon$, while others need restrictions, such as compactness.

Here and for the remainder of the section, a subscript $n$ on quantities defined
in Section~\ref{sec:model} makes dependence on the sample size explicit, and
$\pr_\psi$ is the probability when $V^\tsp Y \sim \rN(0, \tilde{\Sigma}(\psi))$.
Importantly, the sample size $n$ is the length of the vector $Y$ but, in
general, it need not be related to a number of independent observations. Indeed,
the results do not explicitly require any particular type of dependence or lack
thereof.

Pointwise versions of \eqref{eq:uniform_coverage}, which omit the supremum,
often follow from classical arguments for interior parameters
\citep{jiang1996reml}, even if $C_n$ were based on a Wald or likelihood ratio
statistic. However, as discussed in Section \ref{sec:intro} and illustrated in simulations, pointwise asymptotics
are often misleading while uniform asymptotics are more accurately reflected in
finite samples. The following lemma lets us establish
\eqref{eq:uniform_coverage} by verifying conditions along sequences of (true)
parameters. For the remainder of the section, for each $n$, $\psi_n \in
\Upsilon$ denotes the true parameter with interest parameter $\psi_{n1} \in
\R{}$ and nuisance parameter $\psi_{n(2)} \in \R{r - 1}$. Throughout, $r$
does not depend on $n$.

\begin{lemma} \label{lem:sequences}
  If, for every sequence $(\psi_n)$ in $A$, $\lim_{n\to \infty}\pr_{\psi_n}\{\psi_{n1} \in C_n(\alpha)\} = 1 - \alpha$,
  then \eqref{eq:uniform_coverage} holds.
\end{lemma}

The choice of nuisance parameter estimator is important for the coverage of the
profile intervals. To establish \eqref{eq:uniform_coverage}, we will require
that, with probability tending to one, the nuisance estimator is a root of its
score function. That is, with probability tending to one,
\begin{equation} \label{eq:nuisance_root}
  S_{n(2)}(\psi_{n1}, \tilde{\psi}_{n(2)}(\psi_{n1})) = 0.
\end{equation}
Equation \eqref{eq:nuisance_root} is satisfied for every fixed $n$ when
$\tilde{\psi}_{n(2)}(\psi_{n1})$ is a maximizer of the restricted likelihood
over the extended set $\tilde{\Upsilon}_n^{(2)}(\psi_{n1})$. More generally, any
root of the nuisance score is admissible, whether or not it is a maximizer. This
matters because, in practice, convergence of iterative algorithms to a root is
verifiable and occurs under weak conditions, but it is in general unknown
whether the root is a maximizer. Requiring that \eqref{eq:nuisance_root} hold only with
probability tending to one accommodates samples for which no root exists.
The one-step estimators discussed in Section~\ref{sec:model} typically satisfy
\eqref{eq:nuisance_root} only approximately, and are treated separately below.

Our main results also require the nuisance estimator at the true interest parameter, $\tilde{\psi}_{n(2)} = \tilde{\psi}_{n(2)}(\psi_{n1})$, to satisfy
\begin{equation} \label{eq:nuis_rate}
  \Vert \mathcal{I}_{n(22)}(\psi_n)^{1/2}(\tilde{\psi}_{n(2)} - \psi_{n(2)})\Vert = O_p(1).
\end{equation}
When the smallest eigenvalue of $\mathcal{I}_{n(22)}(\psi_n)$ diverges,
\eqref{eq:nuis_rate} says the nuisance estimator is consistent at the
information rate.

To state the main result of the section, we require some more notation. For $v
\in \R{r}$, let $\tilde{\Sigma}(v) = V^\tsp\{Z\Psi(v)Z^\tsp + v_r I\}V$. When $v
= \psi_n$, this is the covariance matrix of $V^\tsp Y$. In general, however,
$\tilde{\Sigma}(v)$ need not be a covariance matrix since any $v \in \R{r}$ is
allowed. Define also, for any $\psi \in \tilde{\Upsilon}_n$, $A_j(\psi) =
\tilde{\Sigma}(\psi)^{-1/2}\{\partial\tilde{\Sigma}(\psi)/\partial\psi_j\}
\tilde{\Sigma}(\psi)^{-1/2}$ and $A(v, \psi) = \sum_{j} v_j A_j(\psi) =
\tilde{\Sigma}(\psi)^{-1/2}\tilde{\Sigma}(v)\tilde{\Sigma}(\psi)^{-1/2}$. The
matrices are related to the expected information as $\mathcal{I}_{n,ij}(\psi) =
\tr\{A_i(\psi)A_j(\psi)\}/2$, and $v^\tsp \mathcal{I}_{n}(\psi)v = \tr\{A(v,
\psi)^2\}/2$. It follows that the Fisher information is positive definite at $\psi$ if and only if $\Vert A(v, \psi)\Vert_F^2  > 0$, where $\Vert \cdot \Vert_F$ is the Frobenius norm. For any such $v$ and $\psi$, let $a_n(v, \psi) =\Vert A(v, \psi)\Vert / \Vert A(v, \psi)\Vert_F$; for every other $v$ and $\psi$ set $a_n(v, \psi) = 1$. Here, $\Vert \cdot \Vert$ is the spectral norm. Finally, let $\bar{a}_n(\psi) = \sup_{v: \Vert v\Vert = 1} a_n(v, \psi)$. Like \eqref{eq:res_ll}, $a_n$ does not depend on the choice of $V$: replacing $V$ by $VO$ for an orthogonal $O$ replaces $A(v, \psi)$ by $O^\tsp A(v, \psi)O$, which has the same spectral and Frobenius norms.

For every $n$, $\bar{a}_n(\psi)$ bounds how far the standardized score
$\mathcal{I}_n(\psi)^{-1/2}S_n(\psi)$ is from normal
\citep[Lemma~5]{ekvall2026uniform}. The main result requires that
\begin{equation} \label{eq:negligible}
  \lim_{n \to \infty}\bar{a}_n(\psi_n)  = 0.
\end{equation}
Under \eqref{eq:negligible}, therefore,
$\mathcal{I}_n(\psi_n)^{-1/2}S_n(\psi_n)$ converges in distribution to $\rN(0,
I_r)$ (\webappendix~\ref{sec:lmm_proofs}). Condition \eqref{eq:negligible} has
been shown to hold along sequences approaching the boundary in models with
independent clusters and in a model with crossed random effects
\citep{ekvall2026uniform}. Without nuisance parameters, it alone gives
asymptotically correct uniform coverage of score confidence regions
\citep{ekvall2026uniform}. With nuisance parameters, additional conditions are
needed. \citet{shedden2026profile} give high-level sufficient conditions for
general parametric models, which the proof of Theorem~\ref{thm:porting} in
\webappendix~\ref{sec:lmm_proofs} shows hold under \eqref{eq:negligible}
for nuisance estimators satisfying \eqref{eq:nuisance_root} and
\eqref{eq:nuis_rate}.

\begin{theorem} \label{thm:porting}
  If \eqref{eq:negligible} holds along every sequence $(\psi_n)$ in a set $A$,
  then, along every such sequence, with probability tending to one
  $S_{n(2)}(\psi_{n1}, \cdot)$ has a root satisfying
  \eqref{eq:nuis_rate}, any $\tilde{\psi}_{n(2)}(\psi_{n1})$ satisfying
  \eqref{eq:nuisance_root} and \eqref{eq:nuis_rate} equals that root with
  probability tending to one, and \eqref{eq:uniform_coverage} holds for
  $C_n(\alpha)$ formed with any such $\tilde{\psi}_{n(2)}(\psi_{n1})$.
\end{theorem}

A key step in the proof is showing that the extended set is large enough asymptotically. While an
extended set exists for every fixed $n$, and such sets have been considered
before, the nuisance estimator behaves well only if the set is large on the
information scale. Lemma~\ref{lem:interior} shows that for any
$v$ such that $\tilde{\Sigma}(v)$ is singular, that is, any $v$ on the boundary of the extended set,
$\Vert \mathcal{I}_n(\psi_n)^{1/2}(v - \psi_n)\Vert \geq 1/\{\sqrt{2}\,
\bar{a}_n(\psi_n)\}$, which diverges under \eqref{eq:negligible}.

Because $\tilde{\Upsilon}^{(2)}_n(\psi_{n1})$ is open, a maximizer of the
restricted likelihood over it satisfies \eqref{eq:nuisance_root}, which gives
the following.

\begin{corollary} \label{cor:maximizer}
  If \eqref{eq:negligible} holds along every sequence $(\psi_n)$ in a set $A$
  and, along every such sequence, $\tilde{\psi}_{n(2)}(\psi_{n1})$ is a
  maximizer of $\ell_n(\psi_{n1}, \cdot)$ over
  $\tilde{\Upsilon}^{(2)}_n(\psi_{n1})$ satisfying \eqref{eq:nuis_rate}, then
  \eqref{eq:uniform_coverage} holds.
\end{corollary}

Theorem \ref{thm:porting} requires no bound on the eigenvalues of the correlation matrix
associated with $\mathcal{I}_n(\psi_n)$, or on the size of the efficient
information in comparison to $\mathcal{I}_{n(11)}$; see the discussion of the
efficient score in \webappendix~\ref{sec:lmm_proofs}. When using a one-step estimator for the nuisance parameter, however, the correlation between the scores for interest and nuisance parameters matters.

Recall \eqref{eq:kappa} and let $\kappa_n = \kappa(\psi_n)$. Recall also the
preliminary estimator $\mathring{\psi}_{n(2)}$, which does not depend on
$\psi_1$, and let $\check{\psi}_{n(2)}$ be the one-step update
\eqref{eq:onestep} from it, evaluated at $\psi_{n1}$. The proposition below
requires $\mathring{\psi}_{n(2)}$ to satisfy \eqref{eq:nuis_rate};
Lemma~\ref{lem:reml_start} in \webappendix~\ref{sec:lmm_proofs} shows the
nuisance block of the restricted maximum likelihood estimator does so if the
multiple correlation $\kappa_n$ is asymptotically bounded away from one.

\begin{proposition} \label{prop:onestep}
  If \eqref{eq:negligible} holds, $\mathring{\psi}_{n(2)}$ satisfies
  \eqref{eq:nuis_rate}, and
  \begin{equation} \label{eq:onestep_margin}
    \frac{\kappa_n}{\sqrt{1 - \kappa_n^2}}\,\bar{a}_n \to 0
  \end{equation}
  along every sequence $(\psi_n)$ in $A$, then \eqref{eq:uniform_coverage}
  holds for $C_n(\alpha)$ formed with $\tilde{\psi}_{n(2)} =
  \check{\psi}_{n(2)}$.
\end{proposition}

The nuisance score does not in general vanish at the one-step estimator, which
adds a term to the standardized score. Condition \eqref{eq:onestep_margin}
makes that term tend to zero in probability, so it does not affect the
limiting distribution. When $\kappa_n$ is bounded
away from one, \eqref{eq:onestep_margin} reduces to \eqref{eq:negligible}.
Although $\kappa_n$ and $\bar{a}_n$ depend on the unknown $\psi_n$, $\kappa_n$
can be estimated by evaluating $\kappa$ at an estimate of $\psi_n$, and
$\bar{a}_n$ can often be bounded uniformly in the parameter
\citep{ekvall2026uniform}.

Example \ref{ex:common_design} verifies the conditions of
Corollary~\ref{cor:maximizer} in a simple mixed model, where a maximizer over
the extended set exists almost surely for every $n$.

\begin{example} \label{ex:common_design}
Suppose the $n = mk$ observations form $m$ clusters of size $k$ with a common
design: for $i \in \{1, \dots, m\}$,
\[
  Y_i = X_1\beta + Z_1U_i + E_i,
\]
where $X_1 \in \R{k\times p}$ has full column rank, $Z_1 \in \R{k\times q_1}$,
and the $U_i \sim \rN_{q_1}(0, \Psi_1)$ and $E_i \sim \rN_k(0, \psi_r I_k)$
are independent. This is \eqref{eq:lmm} with $X = 1_m \otimes X_1$, $Z = I_m
\otimes Z_1$, and $\Psi = I_m \otimes \Psi_1$, where $\otimes$ is the Kronecker
product and $\Psi_1 = \Psi_1(\psi) = \sum_{j = 1}^{r - 1}\psi_j H_{1j}$ for
known symmetric $H_{1j} \in \R{q_1 \times q_1}$. For $v \in \R{r}$, let
$\Sigma_1(v) = Z_1\Psi_1(v)Z_1^\tsp + v_r I_k$, so that $\Sigma(v) = I_m
\otimes \Sigma_1(v)$. The case $Z_1 = X_1 = [1_k, t]$ is a random intercept
and slope with common observation times $t \in \R{k}$. Asymptotics are as $m \to
\infty$ with $k$, $p$, and $r$ fixed, so $n \to \infty$.

\begin{proposition} \label{prop:common_design}
  If $k$, $p$, and $r$ are fixed, $m \geq k + 1$, and $\Sigma_1(v) \neq 0$ for
  every non-zero $v \in \R{r}$, then $\bar{a}_n(\psi) \leq (m - 1)^{-1/2}$ for
  every $\psi \in \tilde{\Upsilon}_n$ and, almost surely,
  $\ell_n(\psi_1, \cdot)$ has a maximizer on
  $\tilde{\Upsilon}^{(2)}_n(\psi_1)$ for every $\psi_1$ such that
  $\tilde{\Upsilon}^{(2)}_n(\psi_1)$ is non-empty. If, in addition,
  $\tilde{\psi}_{n(2)}(\psi_1)$ is such a maximizer for every $n$ and
  $\psi_1$, then \eqref{eq:nuis_rate} holds along every sequence $(\psi_n)$ in
  $\Upsilon$, and hence \eqref{eq:uniform_coverage} holds with $A = \Upsilon$.
\end{proposition}
\end{example}

% simulations.tex

\section{Simulations} \label{sec:simulations}

We study finite-sample properties of the proposed intervals and compare them to
Wald, likelihood ratio, and constrained score intervals. The latter result from
using the smaller parameter set for the nuisance parameters; the interest
parameter has the same set for all methods. All intervals are based on the
restricted likelihood. In every simulation setting, the true parameter is
interior and, therefore, each interval is calibrated against the
$\chi^2_1$-distribution. The simulations show how the behavior of the intervals
depends on the proximity of the parameter to the boundary and other model
quantities. Code reproducing all results is available at
\texttt{https://github.com/koekvall/reconf-suppl}.

The Wald statistic is the squared difference between the estimate and the null
value, times the efficient observed information. Near the boundary the
constrained estimate need not be a stationary point and that information can be
negative, in which case the usual statistic is undefined and the replication is
lost. We use the absolute value instead, so that every replication contributes.
The choice favors the Wald intervals, whose coverage is worse without it.

We first consider coverage probabilities for the covariance between
a random intercept and slope. Data were generated as $Y_{ij} = X_{ij}^\tsp
\beta + Z_{ij}^\tsp U_i + E_{ij}$ for $n_1$ groups indexed by $i$ and $n_2$
observations per group indexed by $j$, so $n = n_1 n_2$, where $X_{ij} = Z_{ij}
= [1, x_{ij}]^\tsp$ and the $U_i \sim \rN_2(0, \Psi_1)$ and $E_{ij} \sim \rN(0,
\psi_4)$ are independent, with
\[
  \Psi_1 = \begin{bmatrix} \psi_1 & \psi_2 \\ \psi_2 & \psi_3 \end{bmatrix}.
\]
The predictors were drawn from the uniform distribution on $(-1, 1)$,
independently across $i$ and $j$; the restricted likelihood does not depend on $\beta$, so it can be zero without loss of generality. In the
notation of Section~\ref{sec:model}, $q = 2n_1$, $r = 4$, and $\Psi(\psi) =
I_{n_1} \otimes \Psi_1$, so $\Upsilon$ is the set of $\psi$ with
$\psi_1 \geq 0$, $\psi_3 \geq 0$, $\psi_1 \psi_3 \geq \psi_2^2$, and $\psi_4 >
0$. The parameter of interest is the covariance $\psi_2$, and the error
variance is $\psi_4 = 1$.

Figure \ref{fig:routes} shows how coverage changes as the variances $\psi_1$ and
$\psi_3$, here set to be the same, approach the boundary (panel (a)) and as the
correlation $\psi_2 / \sqrt{\psi_1 \psi_3}$ approaches one (panel (b)). The
proposed interval is near-nominal across parameter values. The other intervals
are sometimes conservative and sometimes anti-conservative. For example, across
the two panels the Wald interval covers between $0.902$ and $0.995$, the
likelihood ratio between $0.937$ and $0.992$, and the constrained score interval
between $0.448$ and $0.976$, against the proposed interval's $0.945$ to
$0.956$.

Recall that the two score intervals differ only in the set on which the nuisance
parameter is estimated. On the smaller set that estimate can lie on the
boundary, where the nuisance score need not vanish and the standardization in
\eqref{eq:standardized_score} no longer removes the effect of estimating it. In
the case where the estimate lies on the boundary, the
interval can then be empty, and it was in $54\%$ of replications at correlation
$0.999$.

\begin{figure}[ht]
    \includegraphics[width=\linewidth]{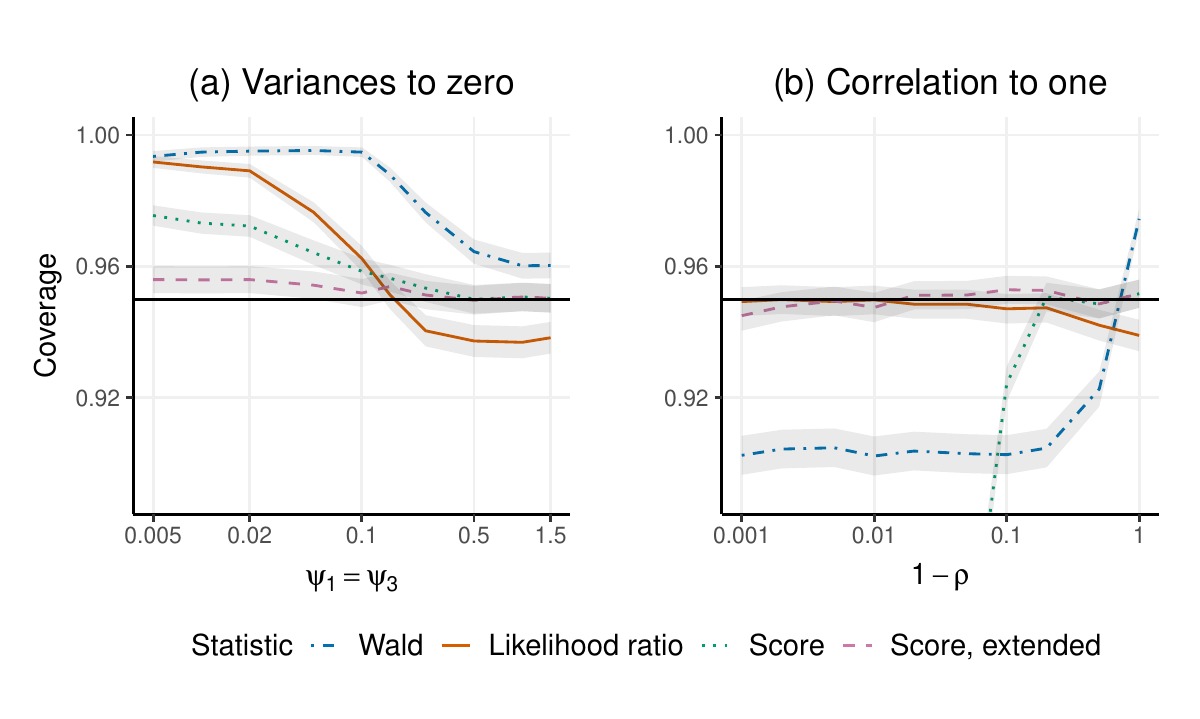}
    \caption{Coverage of nominal $95\%$ intervals for the covariance of a
    correlated random intercept and slope as the parameter approaches the
    boundary (a) $20$ groups of $10$
    observations, $\psi = (\psi_1, 0.2\psi_1, \psi_1, 1)$. (b) $20$ groups of $10$
    observations, $\psi = (2.5, 2.5\rho, 2.5, 1)$, the correlation approaching
    one with the variances held fixed; the horizontal axis is $1 - \rho$.
    Both horizontal axes are logarithmic. The vertical axis is cut, since the
    constrained score interval's coverage near correlation one would otherwise
    obscure the comparison for the other intervals. The solid line is the
    nominal level and the bands are two standard errors from $10\,000$
    replications.}
    \label{fig:routes}
\end{figure}

\begin{figure}[ht]
    \includegraphics[width=0.75\linewidth]{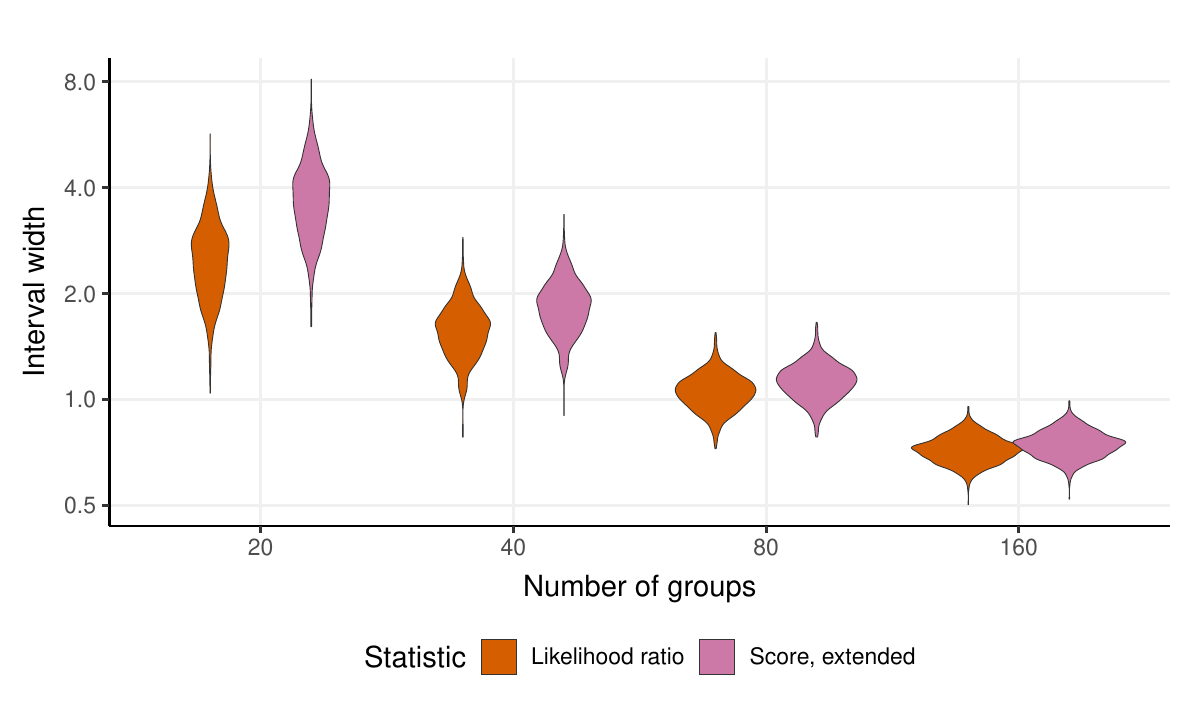}
    \caption{Widths of the likelihood-ratio interval and the
    proposed interval in the correlated design $\psi = (2, 0.4, 2, 1)$
    with $20$, $40$, $80$, or $160$ groups of $10$ observations, from $1000$
    replications; the vertical axis is logarithmic.}
    \label{fig:width}
\end{figure}

% P6c -- width
Figure~\ref{fig:width} compares widths of the proposed intervals and the likelihood ratio ones. The setting is chosen so that both have
near-nominal coverage probability, in the correlated design from Fig.~\ref{fig:routes} with $\psi = (2, 0.4, 2, 1)$ and $n_1 = 20$, $40$, $80$, or $160$ groups. There is substantial overlap between the width distributions, but the likelihood ratio is clearly shorter on average with $n_1 = 20$, and somewhat shorter when $n_1 = 40$. As the sample size increases, the width difference decreases. The median ratios of lengths, paired by replication, were $1.48$, $1.16$, $1.07$, and $1.04$ for $n_1$ equal to $20$, $40$, $80$, and $160$, respectively.

\begin{figure}[ht]
    \includegraphics[width=\linewidth]{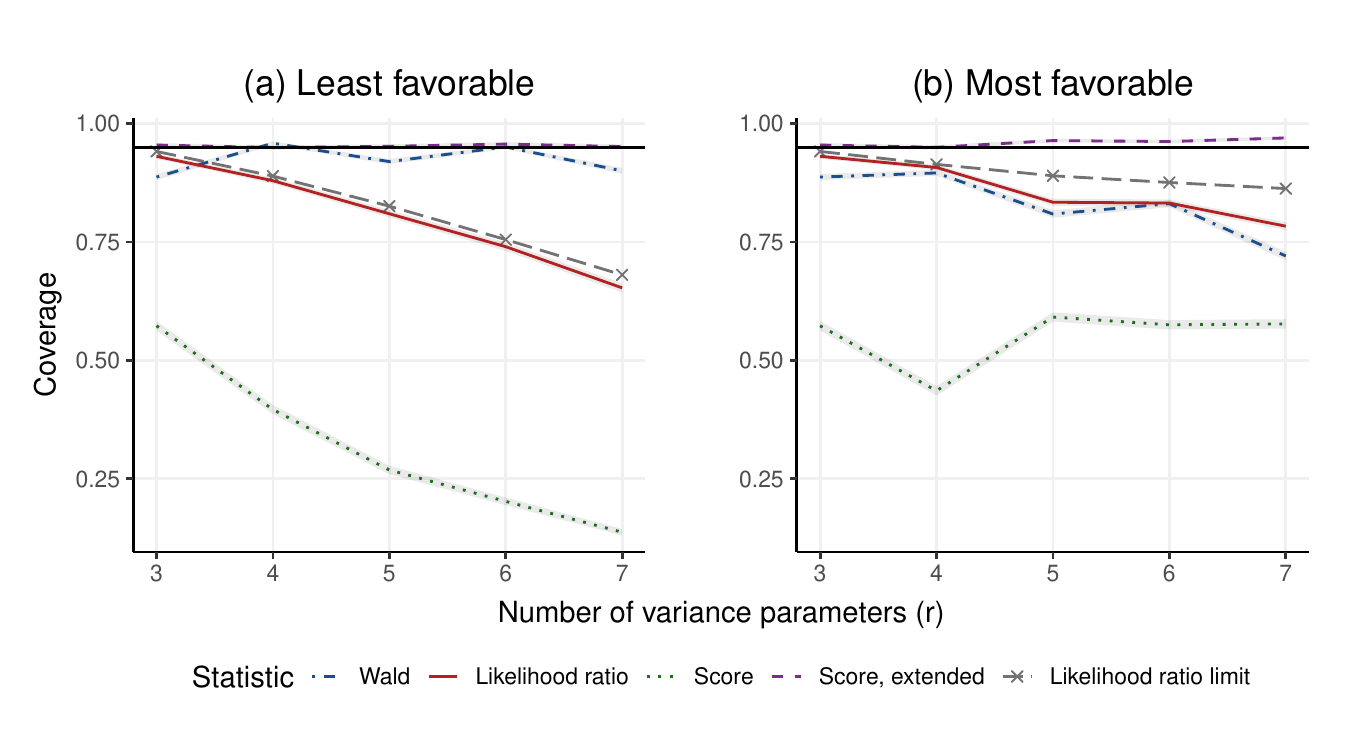}
    \caption{Coverage of nominal $95\%$ intervals for $\psi_1$ in a variance
    components model with $r$ parameters (horizontal axis). The likelihood
    ratio limit, in grey with crosses, is computed from the local
    approximation and does not depend on $n$.}
    \label{fig:boundary_k}
\end{figure}

% P6 -- Figure 2, the mechanism
Figure~\ref{fig:boundary_k} shows how the number of nuisance parameters can, in
designs with correlated score elements and near-boundary parameters, seriously
affect coverage of common intervals. The model is a variance components model,
$Y \sim \rN(0, \sum_{j = 1}^r \psi_j K_j)$, with $K_r = I_n$ the identity
matrix. For simplicity, we constructed $K_1, \ldots, K_{r-1}$ to be diagonal
matrices such that the multiple correlation between the interest and nuisance
scores is $\kappa_n = 0.99$ at every $r$, so the nuisance parameters inflate the
variance for the interest parameter by $(1 - \kappa_n^2)^{-1} = 50$; this is
similar to the data example in Section \ref{subsec:autism}. We also constrained
every nuisance score element to be equally correlated with the interest score,
and with every other nuisance score, after eliminating the error variance as in
the definition of $\kappa_n$. Among all such designs, we searched for the ones
giving the smallest (panel (a)) and greatest (panel (b)) likelihood ratio
limit, the coverage of the likelihood ratio interval when the log-likelihood is
quadratic and the parameter set a cone. Both panels show that limit together
with the simulated coverages. The interval is
for $\psi_1$, with $\psi_1 = \psi_r = 1$ and $\psi_2 = \dots = \psi_{r - 1} =
0.001$, so the nuisance components are near their boundary but interior, and $n
= 800$ at every $r$.

The coverage of the likelihood ratio interval decreases from about $0.93$ at $r
= 3$ to $0.65$ at $r = 7$ in panel (a), and to about $0.78$ in panel (b). Wald
intervals have near nominal coverage in panel (a), but are anti-conservative in
panel (b), reaching below 0.75 at worst. Other structures give other values at
the same $\kappa_n$, depending on the whole geometry and on the unknown
parameter. In practice, therefore, it is in general impossible to
calibrate Wald and likelihood ratio intervals correctly.

% data_examples_autism.tex

\section{Data examples} \label{sec:data_examples}

\subsection{Socialization in children with autism} \label{subsec:autism}

The Vineland Socialization Age Equivalent score measures
parent-reported socialization. \citet{anderson2009patterns} recorded it at ages
$2$, $3$, $5$, $9$, and $13$ years in a longitudinal study of children with
autism spectrum disorders. We use the $610$ observed scores on $158$ children
distributed with the \texttt{WWGbook} package \citep{west2014linear}.
Socialization can improve with age at a rate that differs between children, so we
fit a random intercept and a random slope in age per child,
\begin{align} \label{eq:autism}
  Y_{it} = \beta_1 + \beta_2 X_{it} + \beta_3^\tsp G_i + U_{1i} + U_{2i}X_{it} +
  E_{it},
\end{align}
where $Y_{it}$ is the score for child $i$ at age $X_{it}$, $G_i$ indicates the
child's expressive language group at age two, and the two random effects are
correlated within child. This is \eqref{eq:lmm} with $\Psi =
\bdiag(\Psi_1, \dots, \Psi_{158})$ and $\Psi_i$ the $2\times 2$ matrix with
diagonal $(\psi_1, \psi_3)$ and off-diagonal $\psi_2$. The covariance $\psi_2$
is of interest in its own right: it says whether children who start higher improve faster or more slowly, and by how much.

Because the spread of the response grows with its level, as panel~(a) of
Figure~\ref{fig:autism} shows, we take the error standard deviation
proportional to age and fit \eqref{eq:autism} with known weights $w_{it} =
X_{it}^{-2}$, which is \eqref{eq:lmm} after multiplying $Y$, $X$, and $Z$ by
$W^{1/2}$, so that the theory and the software apply unchanged and $\psi_4$
is the variance of $w_{it}^{1/2}E_{it}$.

The restricted maximum likelihood estimate of the random-effects covariance
is nearly singular: its correlation is $-0.974$, its eigenvalues are $59.5$
and $0.56$, and the direction belonging to the smaller one contributes on
average $0.23$ times the error variance to a weighted observation, against
$1.81$ for the larger. The estimate is thus close to the boundary of
$\Upsilon$ without being on it, suggesting a true parameter near the boundary.
The scores of the components are also strongly correlated: writing
$\kappa_j(\hat{\psi})$ for the multiple correlation \eqref{eq:kappa} with
$\psi_j$ as the interest parameter, evaluated at the estimate $\hat{\psi}$,
$(\kappa_1(\hat{\psi}), \dots, \kappa_4(\hat{\psi})) = (0.990, 0.997, 0.989,
0.489)$, so at that value estimating the other components inflates the
variances of the estimators of $\psi_1$, $\psi_2$, and $\psi_3$ by factors of
$49$, $162$, and $47$ relative to their values when the other components are
known. Most of the uncertainty about $\psi_1$, $\psi_2$, and $\psi_3$ therefore
comes from having to estimate the other components.

Table~\ref{tab:autism} gives $95\%$ intervals for all four components, and
panel~(b) of Figure~\ref{fig:autism} shows the inversion behind the interval for
$\psi_2$, which excludes zero under every method: children who start higher
improve more slowly. The proposed intervals, with the nuisance parameters
partially maximized on the extended set, are finite for every component. We also
include profile intervals from \texttt{lme4}, in the last column, because the
package is widely used. However, that package profiles only the unrestricted
likelihood, even when the fit is restricted. Here three of the four intervals it
returns have an endpoint at zero or infinity, returned when the profiling step
fails rather than computed from the likelihood. The failures illustrate that
settings near the boundary are also difficult computationally. Computing all
four intervals took about $0.5$ seconds with \texttt{reconf} and about $7.0$
seconds with \texttt{lme4} (medians of five runs).

Inverting the restricted likelihood ratio directly, which gives the
profile-likelihood intervals in Table~\ref{tab:autism}, returns an interval for
every component, and the profile-score intervals are between $1$ and $5$ percent
wider than the profile-likelihood intervals. Panel~(b)
also shows the Wald statistic: its interval is symmetric about the estimate by
construction, the shortest of the three, and shifted toward zero relative to the
other two. The extra width of the profile-score intervals may be
read as the price of a uniform validity guarantee the profile-likelihood and
Wald intervals do not have.

\begin{figure}[ht]
    \includegraphics[width=\linewidth]{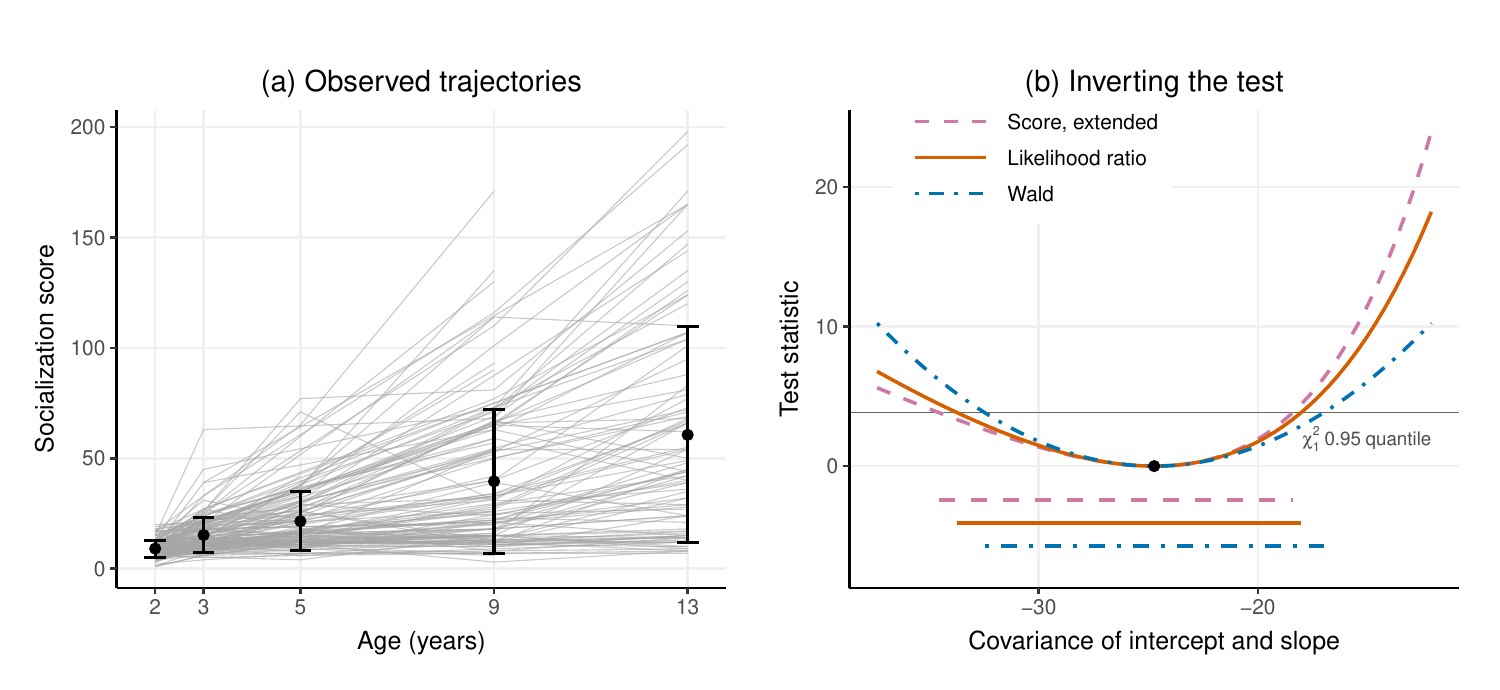}
    \caption{Autism socialization data. (a) Observed trajectories for the $158$
    children, with the mean and one standard deviation at each age. (b) The
    profile score, restricted likelihood ratio, and Wald statistics as functions
    of the intercept--slope covariance, with the $0.95$ quantile of the
    $\chi^2_1$-distribution; the point marks the restricted maximum likelihood
    estimate and the bars below give the interval each statistic returns.}
    \label{fig:autism}
\end{figure}

\begin{table}[ht]
\centering
\caption{Estimates and 95\% confidence intervals for the autism data}
\footnotesize\setlength{\tabcolsep}{4pt}
\begin{tabular}{llrccc}
\toprule
  & & & \multicolumn{3}{c}{95\% Confidence Interval} \\
  \cmidrule(lr){4-6}
  Component & & REML & Profile score & Profile likelihood & \texttt{lme4} \\
\midrule
  $\psi_1$ & Intercept variance & 46.089 & $(31.465,\, 68.464)$ & $(30.887,\, 66.256)$ & $(0.000,\, \infty)$ \\
  $\psi_2$ & Covariance & $-24.727$ & $(-34.566,\, -18.401)$ & $(-33.737,\, -18.018)$ & $(-\infty,\, -23.631)$ \\
  $\psi_3$ & Slope variance & 13.993 & $(10.941,\, 18.759)$ & $(10.712,\, 18.409)$ & $(12.885,\, \infty)$ \\
  $\psi_4$ & Error variance & 2.521 & $(2.176,\, 2.994)$ & $(2.159,\, 2.966)$ & $(2.163,\, 2.948)$ \\
\bottomrule
\end{tabular}
\label{tab:autism}
\end{table}

% data_examples_snp.tex

\subsection{Genetic variance components in mice} \label{subsec:snp}

A single-nucleotide polymorphism (SNP) is a position in the genome at which
the alleles observed in a population differ in one nucleotide; the SNPs
studied here have two alleles. In mice, the genotype at a SNP is the number
of copies, $0$, $1$, or $2$, of the less common allele. The frequency of
that allele among the genotyped individuals is the SNP's minor allele
frequency. For $n$ individuals genotyped at $s$ SNPs, let $D \in \R{n \times s}$ be
the matrix of genotypes, with the column for a SNP of minor allele frequency
$\hat{p}$ centered at its sample mean $2\hat{p}$ and scaled by
$\{2\hat{p}(1 - \hat{p})\}^{1/2}$; missing genotypes, $3\%$ of the total,
are set to the sample mean before centering. Partition the SNPs into groups,
and let $D_j$ hold the columns of $D$ in group $j$ and $s_j$ their number.
The genetic relationship matrix of group $j$, $K_j = c_jD_jD_j^\tsp/s_j$
with $c_j$ chosen so that the diagonal elements of $K_j$ average one,
measures the genetic similarity of every pair of individuals through the
SNPs in that group; its entries are proportional to estimates of twice the
kinship coefficients of quantitative genetics. Grouping by
minor allele frequency separates the contributions of rare and common
variants \citep{yang2015}.

We illustrate the method on white blood cell count in $742$ Diversity
Outbred mice, a population descended from eight inbred founder strains, using the
\texttt{DO\_Gatti2014} data in the \texttt{qtl2data} repository, which holds
data in the format of the \texttt{qtl2} package. The $6044$ autosomal SNPs are
divided into four groups by minor allele frequency, with cut points $0.05$,
$0.15$, and $0.3$. We use four groups rather than the seven of
\citet{yang2015} because the SNPs are mostly common: the median minor
allele frequency is $0.38$, and only $16$ SNPs fall below $0.05$. The
fixed effects are sex, birth date, and the number of
generations of outbreeding since the population was founded. The model is
\eqref{eq:lmm} with $Z = [c_1^{1/2}D_1/s_1^{1/2}, \dots, c_4^{1/2}D_4/s_4^{1/2}]$ and $\Psi$
block diagonal with blocks $\psi_jI_{s_j}$, so that $r = 5$ and $\Sigma =
\sum_{j = 1}^4\psi_jK_j + \psi_5I_n$: every SNP has an independent random
effect, with a variance that depends only on its group. Table~\ref{tab:snp} lists the groups
and the number of SNPs $s_j$ in each. The questions in such an analysis
are whether a class of variants contributes to the trait at all, and how
large the components are relative to each other, which ranks the classes by their contribution.

The restricted maximum likelihood estimate of $\psi_2$ is zero, suggesting a true parameter near or at the boundary.
The scores of the components are correlated: with $\psi_4$ as the
interest parameter, the multiple correlation \eqref{eq:kappa} between its
score and the nuisance scores, evaluated at the restricted maximum likelihood
estimate $\hat{\psi}$, is $\kappa(\hat{\psi}) = 0.74$, so at that value the
nuisance parameters inflate the variance of the estimator of $\psi_4$ by the
factor $\{1 - \kappa(\hat{\psi})^2\}^{-1} = 2.2$.

Table~\ref{tab:snp} gives $95\%$ intervals for the five components: the proposed
profile-score intervals, profile-likelihood intervals as in
Section~\ref{subsec:autism}, and the universal inference intervals of
\citet{zhang2026universal}. We include no other software in the comparison
because we know of none that computes restricted profile-likelihood intervals in
this setting: unlike the model of Section~\ref{subsec:autism}, this one cannot
be fitted by \texttt{lme4} without extensions, and its profile intervals are
based on the unrestricted likelihood even for restricted fits, as noted in
Section~\ref{subsec:autism}. The universal intervals were computed with the
\texttt{varcomp} package of those authors after removing the fixed effects by an
orthogonal transformation, which gives the restricted likelihood. The split
likelihood ratio was averaged over ten random splits of the data into halves,
which preserves validity \citep{wasserman2020universal}, because the interval
from a single split is valid but unstable: over the ten splits, the upper
endpoint for the component with minor allele frequency in $[0.05, 0.15)$ ranged
from $0.9$ to $3.3$.

The universal intervals, which are uniformly valid in finite samples, are $1.7$ to $3.1$
times wider than the profile-score intervals for four of the five components.
For the rare-variant component the universal interval is unbounded: the
information about it comes only from the $16$ directions spanned by its
relationship matrix. The profile-score intervals are
$0.99$ to $1.8$ times as wide as the profile-likelihood intervals, which have no
asymptotic coverage guarantee with a nuisance parameter at or near the
boundary.

The profile-score and profile-likelihood intervals agree on which components
differ significantly from zero, $\psi_4$ and $\psi_5$, whereas the universal
intervals exclude zero only for $\psi_5$, the error variance. Computing the five
intervals took $82$ seconds with the profile score, $219$ with the profile
likelihood, and $3865$ with universal inference over the ten splits, each timed
once.

Interest is sometimes in the ratio of a component to the total variance, the
heritability attributable to a class of variants. Testing whether such a ratio is zero is equivalent to
testing whether the component is zero, and ranking the classes by their
ratios, as is common in analyses of partitioned heritability
\citep{yang2015}, is equivalent to ranking them by their components.
Intervals for the ratios can in principle be obtained by a
reparameterization, which is not pursued here.

\begin{table}[ht]
\centering
\caption{Estimates and 95\% confidence intervals for the mouse data}
\footnotesize
\begin{tabular}{llrrccc}
\toprule
  & & & & \multicolumn{3}{c}{95\% Confidence Interval} \\
  \cmidrule(lr){5-7}
  Component & Frequency & SNPs & REML & Profile score & Profile likelihood & Universal \\
\midrule
  $\psi_1$ & $[0,\,0.05)$ & 16 & 0.006 & $(0.000,\; 0.241)$ & $(0.000,\; 0.131)$ & $(0.000,\; \infty)$ \\
  $\psi_2$ & $[0.05,\,0.15)$ & 317 & 0.000 & $(0.000,\; 0.332)$ & $(0.000,\; 0.337)$ & $(0.000,\; 1.017)$ \\
  $\psi_3$ & $[0.15,\,0.3)$ & 920 & 0.312 & $(0.000,\; 1.310)$ & $(0.000,\; 1.049)$ & $(0.000,\; 2.261)$ \\
  $\psi_4$ & $[0.3,\,0.5]$ & 4791 & 1.294 & $(0.520,\; 2.312)$ & $(0.528,\; 2.215)$ & $(0.000,\; 3.689)$ \\
  $\psi_5$ & Error & --- & 3.629 & $(3.054,\; 4.323)$ & $(3.061,\; 4.269)$ & $(2.529,\; 5.042)$ \\
\bottomrule
\end{tabular}
\label{tab:snp}
\end{table}

% discussion.tex

\section{Discussion} \label{sec:discussion}

Coverage of common confidence intervals depends on the proximity of the
(unknown) true parameter to the boundary, and more generally on the geometry of
the parameter set near the true parameter. The structure of the information
matrix at the truth also has an effect. The proposed intervals, by contrast,
have asymptotically correct coverage uniformly in the parameter and near-nominal
coverage in every simulated setting, and we therefore recommend them as the
default. Likelihood ratio intervals are typically shorter, but how to modify
them to ensure uniform asymptotic validity in near-boundary settings is an open
question. Universal inference intervals are valid in finite samples, but they
can be both wide and computationally expensive.

The key condition of our theory holds in several linear mixed model settings, including independent clusters and crossed designs. In
general, no particular type of dependence or lack thereof is required. We expect
the results to extend to other normal models. The setting here covers any covariance matrix linear in the parameters, but nonlinear ones, such as autoregressive, Mat\'ern, or factor-analytic covariances, require additional work. Extending the results to generalized linear mixed
models is more difficult: the likelihood integrates against the density of the
random effects, so it needs $Z\Psi(\psi)Z^\tsp$ to be a covariance matrix and
the extended set of Section~\ref{sec:model} has no clear analogue. 

Other open questions include intervals for ratios,
simultaneous coverage for several components, and settings with high-dimensional
nuisance parameters. With dense relationship matrices each evaluation of the
statistic costs $O(rn^3 + r^2n^2)$, which limits the sample sizes reachable in
genetics and may require additional computational developments.

\bibliographystyle{apalike}
\bibliography{compute_lmm}

\appendix
\renewcommand{\appendixname}{Appendix}

\numberwithin{theorem}{section}
\numberwithin{lemma}{section}
\numberwithin{corollary}{section}
\numberwithin{definition}{section}
\numberwithin{example}{section}

% app_lmm_proofs.tex

\section{Proofs of the asymptotic coverage results} \label{sec:lmm_proofs}

Throughout, we fix a sequence $(\psi_n)$ in $\Upsilon$, let
$\bar{a}_n =
\bar{a}_n(\psi_n)$, $R = \tilde{\Sigma}(\psi_n)^{-1/2}V^\tsp Y$, which is
$\rN(0, I_{n - p})$ under $\psi_n$, $\mathcal{J}_n(\psi) = -\partial S_n(\psi)/\partial\psi^\tsp$ is the observed information, $\mathbb{S}^{r - 1} = \{u \in \R{r} : \Vert u\Vert = 1\}$, and $P_2 = [0, I_{r-1}] \in \R{(r-1)\times r}$, such that $\psi_{(2)} = P_2\psi$. 

Furthermore, under \eqref{eq:negligible}, $\mathcal{I}_n(\psi) \succ 0$ for every $\psi \in \tilde{\Upsilon}_n$ and large enough $n$. The justification is as follows: note that $v^\tsp \mathcal{I}_n(\psi) v = \| A(v, \psi) ||_F^2$, and thus $\mathcal{I}_n(\psi) \succ 0$ if and only if $A(v, \psi) \neq 0$ for every nonzero $v$. $a_n(v, \psi_n) = 1$ whenever $\Vert A(v, \psi_n)\Vert_F = 0$, which contradicts \eqref{eq:negligible} for large $n$, and thus $\| A(v, \psi_n) \| \neq 0$ for large $n$ and nonzero $v$. Since $A(v, \psi_n) = \tilde{\Sigma}(\psi_n)^{-1/2}\tilde{\Sigma}(v)\tilde{\Sigma}(\psi_n)^{-1/2}$ with $\tilde{\Sigma}(\psi_n)^{-1/2}$ invertible, $A(v, \psi_n) = 0$ if and only if $\tilde{\Sigma}(v) = 0$, and thus $\tilde{\Sigma}(v) \neq 0$ for large $n$ and nonzero $v$. Thus, $A(v, \psi) = \tilde{\Sigma}(\psi)^{-1/2}\tilde{\Sigma}(v)\tilde{\Sigma}(\psi)^{-1/2} \neq 0$ for every non-zero $v$, giving the result.

For $M > 0$, let
\begin{align}
\label{eq:infoball}
  D_n^{M} &= \{\psi : \Vert \mathcal{I}_n(\psi_n)^{1/2}(\psi - \psi_n)\Vert \leq M\}, \\
  B_n^{M} &= \{\psi_{(2)} : \Vert \mathcal{I}_{n(22)}(\psi_n)^{1/2}(\psi_{(2)} - \psi_{n(2)})\Vert \leq M\},
\end{align}
such that $\Vert \mathcal{I}_n(\psi_n)^{1/2}h\Vert = \Vert
\mathcal{I}_{n(22)}(\psi_n)^{1/2}(\psi_{(2)} - \psi_{n(2)})\Vert$ when $h = (0,
\psi_{(2)}^\tsp - \psi_{n(2)}^\tsp)^\tsp$. By \eqref{eq:negligible}, there exists a sequence $M_n \rightarrow \infty$ with $M_n^3 \bar{a}_n \rightarrow 0$ (take $M_n = \bar{a}_n^{-1/6}$, for example). Such a sequence also satisfies $M_n\bar{a}_n \rightarrow 0$. Fix one such sequence. Lemma~\ref{lem:nuis_est}, Lemma~\ref{lem:reml_start}, and the proof of Theorem~\ref{thm:porting} take $M_n$ to be this fixed sequence and write $D_n = D_n^{M_n}$ and $B_n = B_n^{M_n}$. Lemmas \ref{lem:interior}, \ref{lem:local_ulln}, and \ref{lem:local_detrmn} are stated for an arbitrary radius sequence and thus keep the superscript $M_n$, since the proof of Proposition~\ref{prop:onestep} requires a different sequence, introduced there.

The following technical lemma bounds the supremum of a Gaussian quadratic-form process.

\begin{lemma} \label{lem:qf_sup}
  Let $\xi \sim \rN(0, I_n)$ and let $\{G_t : t \in T\}$ be $n \times n$
  matrices indexed by a nonempty set $T$ contained in a Euclidean ball of radius $\rho > 0$ in
  $\R{k}$. Suppose that, for some $\sigma \geq 0$ and $L > 0$, $\sup_{t \in T}\Vert G_t\Vert_F \leq \sigma$ and $\Vert G_t -
  G_{t'}\Vert_F \leq L\Vert t - t'\Vert$ for all $t, t' \in T$. Then
  \[
    \E \sup_{t \in T} \left\vert \xi^\tsp G_t \xi - \tr G_t \right\vert \leq C_k\,(\sigma + L\rho)
  \]
  for a constant $C_k$ depending only on $k$. In particular, the supremum is
  $O_p(\sigma + L\rho)$.
\end{lemma}

\begin{proof}
  Let $X_t = \xi^\tsp G_t \xi - \tr G_t$ and let $\Vert \cdot \Vert_\phi$ be the
  Orlicz norm for $\phi(x) = e^x - 1$; that is, $\Vert X\Vert_\phi = \inf\{s > 0
  : \E \phi(\vert X\vert/s) \leq 1\}$. The process $(X_t)$ has (Lipschitz) continuous paths 
  since $\vert X_t - X_{t'}\vert \leq (\Vert \xi\Vert^2 + \sqrt{n})\Vert G_t -
  G_{t'}\Vert_F \leq (\Vert \xi\Vert^2 + \sqrt{n}) L \| t - t' \|$, and $T$ has a countable dense subset as a subset of $\mathbb{R}^k$, and thus $\sup_{t \in T} X_t$ and $\sup_{t, t' \in T} |X_t - X_{t'}|$ are measurable.

  Consider an $n\times n$ matrix $G$, random vector $\xi \sim \rN(0, I_n)$ (such that $\E \xi^\tsp G \xi = \tr G$), and $X = \xi^T G \xi - \tr G$. By the Hanson--Wright inequality \citep[Theorem~6.2.1]{vershynin2018high},
  \begin{equation}
  \label{eq:hansen}
    \pr\{\vert X \vert \geq u\} \leq 2\exp\left[-c\min\left(
      \frac{u^2}{\Vert G\Vert_F^2},\ \frac{u}{\Vert G\Vert}\right)\right], \qquad u \geq 0,
  \end{equation}
  for a universal constant $c > 0$. Define $K_1 = \|G\|_F/\min(c, \log 2)$ and consider some $u \geq 0$. When $u \geq \|G\|_F$, $u^2/\|G\|_F^2 \geq u/\|G\|_F$. Furthermore, $u/\|G\| \geq u/\|G\|_F$ always, so by \eqref{eq:hansen},
  $$
  P(|X| \geq u) \leq 2 \exp \left[ \frac{-cu}{||G||_F} \right] \leq 2\exp(-u/K_1)
  $$
  When $u < \| G \|_F$,
  $$
  2 \exp(-u/K_1) > 2 \cdot \frac{1}{2} = 1 \geq P(|X| \geq u).
  $$
  Therefore, for every $u \geq 0$, $P(|X| \geq u) \leq 2 \exp(-u/K_1)$. By \citep[Proposition~2.7.1]{vershynin2018high}, his (a) with $K_1$ implies (d): there exists $K_2 \leq C_0 K_1$ for some absolute constant $C_0$, with
  $$
  \E \exp(|X|/K_2) \leq 2.
  $$
  By definition, $\|X\|_\phi$ is the infimum over all $K_2$ such that the above holds, and hence
  $$
  \| X \|_\phi \leq K_2 \leq C_0 K_1 = \frac{C_0 ||G||_F}{\min(c, \log 2)} = C ||G||_F
  $$
  for a universal constant $C$.
  
  Applying this with $G = G_{t_0}$
  for a fixed $t_0 \in T$ and with $G = G_t - G_{t'}$ gives
  \[
    \Vert X_{t_0}\Vert_\phi \leq C\sigma, \qquad
    \Vert X_t - X_{t'}\Vert_\phi \leq CL\Vert t - t'\Vert, \quad t, t' \in T.
  \]
  The second condition above means we can apply Theorem 2.2.4 of
  \citet{vandervaart1996weak}. By the justification given at the beginning of this proof, the process $(X_t)$ is separable in the sense of the theorem. Further take their $\psi$ to be $\phi$ and their $c=1$, with $\phi(x)\phi(y)/\phi(xy) \leq e^{x+y}/(e^{xy} - 1) \leq 2e^{-(xy-x-y)} \rightarrow 0$ as $x, y \rightarrow \infty$ (the second inequality holding for $xy \geq \log2$), and $\phi$ is convex and increasing with $\phi(0) = 0$. Lastly, take semi-metric $d(t, t') = CL \Vert t - t'\Vert$, for which their Lipschitz condition holds with constant $1$, and $\delta = \eta = 2CL\rho$, where $2CL\rho$ is an upper bound on the diameter of $T$ under $d$. Then,
  \begin{equation}
  \label{eq:integral}
  \left\| \sup_{t, t' \in T} |X_t - X_{t'} | \right\|_\phi \leq K \left[ \int_0^{2CL\rho} \phi^{-1} \left( D(\epsilon, d) \right) d\epsilon + 2CL\rho \cdot \phi^{-1} \left( D(2CL\rho, d)^2\right) \right] 
  \end{equation}
  where $K$ depends only on $\phi$, and $D(\epsilon, d) = D(T, \epsilon, d)$ is the maximum number of $\epsilon$-separated points in $T$ (the packing number).

  Here $D(2CL\rho, d) = 1$ and $\phi^{-1}(y) = \log(1+y)$, so the second term is $2CL\rho\log(2)$. For the first,
  $$
  D(\epsilon, d) = D(\epsilon/CL, \| \cdot \|) \leq D(B(0, \rho), \epsilon/CL, \| \cdot \|)
  $$
  where $B(0, \rho) \subset \mathbb{R}^k$ is a ball of radius $\rho$. By \citep[Proposition~4.2.12]{vershynin2018high} applied to $K=B(0, \rho)$,
  $$
  \begin{aligned}
  D(B(0, \rho), \epsilon/CL, \| \cdot \|) &\leq \frac{(\rho + \epsilon/(2CL))^k}{(\epsilon/(2CL))^k} \\
  &= \left( \frac{2CL\rho}{\epsilon} + 1 \right)^k \\
  &\leq \left(\frac{6CL \rho}{\epsilon} \right)^k
  \end{aligned}
  $$
  where the last line follows from the fact that $\epsilon \leq 2CL\rho$ in the range of integration, and hence $x := 2CL\rho/\epsilon \geq 1$, and for such $x$, $(x+1)^k \leq (3x)^k$. Since $\phi^{-1}(y) = \log(1+y)$, is an increasing function of $y$, substitution into \eqref{eq:integral} yields
  $$
  \left\| \sup_{t, t' \in T} |X_t - X_{t'} | \right\|_\phi \leq K \left[ \int_0^{2CL\rho} \log \left( 1+\left(\frac{6CL\rho}{\epsilon} \right)^k \right) d\epsilon  + 2CL\rho\log2 \right]. 
  $$
  Substitute $v = \epsilon/(2CL\rho)$ such that the integral becomes
  $$
  2CL\rho \int_0^1 \log\left( 1 + \left( \frac{6}{2v} \right)^k \right) dv = 2CL\rho c_k
  $$
  where $c_k$ is a finite constant by integrability. Thus,
  $$
  \left\| \sup_{t, t' \in T} |X_t - X_{t'} | \right\|_\phi \leq K(2CL\rho c_k + 2CL\rho \log 2) = K_k L \rho
  $$
  where $K_k$ depends only on $k$. Since $\phi(x) \geq x$ implies $\E\vert
  X\vert \leq \Vert X\Vert_\phi$,
  \begin{align*}
    \E\sup_{t, t' \in T}\vert X_t - X_{t'}\vert \leq  K_k L\rho,
  \end{align*}
  
  Lastly, for every $t \in T$, fixing $t_0 \in T$, 
  $$
  |X_t| \leq |X_{t_0}| + |X_t - X_{t_0}| \leq |X_{t_0}| + \sup_{t, t' \in T} |X_t - X_{t'} |.
  $$
  Taking expectations and applying $\E|X_{t_0}| \leq \|X_{t_0}\|_\phi \leq C \sigma$,
  \[
    \E\sup_{t \in T}\vert X_t\vert \leq \E\vert X_{t_0}\vert +
    \E\sup_{t, t' \in T}\vert X_t - X_{t'}\vert \leq C\sigma + K_k L\rho,
  \]
  and taking $C_k := \max(C, K_k)$ gives the result. Markov's inequality gives the stochastic-order statement.
\end{proof}

The extended set $\tilde{\Upsilon}_n$ depends on $n$ through the design matrix $Z \in \mathbb{R}^{n\times q}$ and the projection $V \in \mathbb{R}^{n \times (n-p)}$. 
The results require a nuisance estimator interior to
$\tilde{\Upsilon}_n$, and a set shrinking too fast on the information scale
would preclude one. The following lemma shows that the extended set is
asymptotically large enough.

\begin{lemma} \label{lem:interior}
  If $\psi^0 \in \tilde{\Upsilon}_n$, then every $\psi$ such that
  $\tilde{\Sigma}(\psi)$ is singular satisfies
  \[
    \Vert \mathcal{I}_n(\psi^0)^{1/2}(\psi - \psi^0)\Vert \geq
    \frac{1}{\sqrt{2}\,\bar{a}_n(\psi^0)}.
  \]
  Suppose in addition that \eqref{eq:negligible} holds, and consider $D_n^{M_n}$ as in \eqref{eq:infoball}, with radius sequence $M_n$ satisfying $M_n \bar{a}_n \rightarrow 0$. Then $D_n^{M_n} \subseteq \tilde{\Upsilon}_n$ for all large $n$.
\end{lemma}

\begin{proof}
  Pick an arbitrary $\psi$ such that $\tilde{\Sigma}(\psi)$ is singular and let
  $h = \psi - \psi^0$. Linearity gives $\tilde{\Sigma}(\psi) -
  \tilde{\Sigma}(\psi^0) = \tilde{\Sigma}(h)$, so $\tilde{\Sigma}(\psi) =
  \tilde{\Sigma}(\psi^0)^{1/2}\{I_{n-p} + A(h,
  \psi^0)\}\tilde{\Sigma}(\psi^0)^{1/2}$. Since $\tilde{\Sigma}(\psi^0) \succ 0$
  and $\tilde{\Sigma}(\psi)$ is singular, $-1$ is an eigenvalue of $A(h,
  \psi^0)$, meaning $\Vert A(h, \psi^0)\Vert \geq 1$. 

  Note that $a_n(cv, \psi) = a_n(v, \psi)$ for any constant $c$, and thus
  $$
  a_n(h, \psi^0) = a_n\left( \frac{h}{\|h\|}, \psi^0 \right) \leq \sup_{u : \|u\| = 1} a_n(u, \psi^0) = \bar{a}_n(\psi^0)
  $$
  and thus $\Vert A(h, \psi^0)\Vert_F \geq
  \Vert A(h, \psi^0)\Vert / \bar{a}_n(\psi^0) \geq 1/\bar{a}_n(\psi^0)$, applying $\Vert A(h, \psi^0)\Vert \geq 1$ from above. By the identity $\Vert
  \mathcal{I}_n(\psi^0)^{1/2}h \Vert^2 = h^\tsp \mathcal{I}_n(\psi^0)h =
  \tfrac{1}{2}\Vert A(h, \psi^0)\Vert_F^2$ from Section~\ref{sec:gen_theory}, the bound follows.

  For the second claim, $D_n^{M_n}$ is convex and contains $\psi_n \in
  \tilde{\Upsilon}_n$. If some $\psi' \in D_n^{M_n}$ had $\tilde{\Sigma}(\psi')$ not
  positive definite, the segment $[\psi_n, \psi'] \subseteq D_n^{M_n}$ would contain a
  point with $\tilde{\Sigma}$ singular: either an eigenvalue of
  $\tilde{\Sigma}(\psi')$ is zero, or it has a negative eigenvalue since it is
  not positive semi-definite. In the latter case, let $\psi(t) = \psi_n + t(\psi' - \psi_n)$ for $t \in [0, 1]$. Weyl's inequality ($|\lambda_k(B) - \lambda_k(C)| \leq \| B - C \|$ for symmetric $B, C$ and every $k$) and linearity of $\tilde{\Sigma}(\psi)$ gives that the eigenvalues of $\tilde{\Sigma}(\psi(t))$ are (Lipshitz) continuous functions of $t$. Letting $f(t) = \lambda_{min}(\tilde\Sigma(\psi(t)))$, where $f(0) > 0$ by positive definiteness of $\tilde{\Sigma}(\psi_n)$, and $f(1) < 0$ by $\tilde{\Sigma}(\psi')$ having a negative eigenvalue, it follows by the intermediate value theorem that there exists $t^*$ such that $f(t^*)  = \tilde{\Sigma}(\psi(t^*))) = 0$, where $\psi(t^*)$ is in $D_n^{M_n}$. Combining this with the first part of the proof,
  $$
  M_n \geq \| \mathcal{I}_n(\psi_n)^{1/2} (\psi(t^*) - \psi_n) \| \geq \frac{1}{\sqrt{2}\bar{a}_n(\psi_n)}
  $$
  which is a contradiction upon taking $n$ large enough that $\sqrt{2}M_n\bar{a}_n < 1$. Thus $D_n^{M_n} \subseteq \tilde{\Upsilon}_n$.
\end{proof}

\begin{lemma} \label{lem:local_ulln}
  If \eqref{eq:negligible} holds, then for any radius $M_n$ with $M_n\bar{a}_n \to 0$ and $\liminf_n M_n > 0$,
  \[
    \sup_{\psi \in D_n^{M_n}} \left\Vert \mathcal{I}_n(\psi_n)^{-1/2}\{\mathcal{J}_n(\psi) - \mathcal{I}_n(\psi_n)\}\mathcal{I}_n(\psi_n)^{-1/2}\right\Vert = O_p(M_n\bar{a}_n).
  \]
\end{lemma}

\begin{proof}
  We first bound, pointwise in $(\psi, u) \in D_n^{M_n} \times \mathbb{S}^{r - 1}$,
  \[
     \left\vert u^\tsp \mathcal{I}_n(\psi_n)^{-1/2}\{\mathcal{J}_n(\psi) - \mathcal{I}_n(\psi_n)\}\mathcal{I}_n(\psi_n)^{-1/2}u\right\vert = \left \vert w^\tsp\{\mathcal{J}_n(\psi) - \mathcal{I}_n(\psi_n)\} w\right\vert,
  \]
  where $w = \mathcal{I}_n(\psi_n)^{-1/2}u$, and then conclude by the chaining
  bound of Lemma~\ref{lem:qf_sup}. Fix $n$ large
  enough that $D_n^{M_n} \subseteq \tilde{\Upsilon}_n$ by
  Lemma~\ref{lem:interior}. Pick arbitrary $u \in \mathbb{S}^{r - 1}$
  and $\psi \in D_n^{M_n}$, so $\tilde{\Sigma}(\psi) \succ 0$. Let $h = \psi - \psi_n$
  and $\Phi =
  \tilde{\Sigma}(\psi_n)^{1/2}\tilde{\Sigma}(\psi)^{-1}\tilde{\Sigma}(\psi_n)^{1/2}$.

  Linearity of $\tilde{\Sigma}$ gives $\tilde{\Sigma}(\psi) -
  \tilde{\Sigma}(\psi_n) = \tilde{\Sigma}(h)$, and hence
  \[
    \Phi^{-1} = \tilde{\Sigma}(\psi_n)^{-1/2}\tilde{\Sigma}(\psi)\tilde{\Sigma}(\psi_n)^{-1/2} = I_{n-p} + A(h, \psi_n).
  \]
  A direct calculation shows that the observed information has entries
  \[
    \mathcal{J}_{n,ij}(\psi) = R^\tsp \Phi A_i \Phi A_j \Phi R - \tfrac{1}{2}\tr(\Phi A_i \Phi A_j).
  \] 
  Thus, since $\mathcal{I}_{n,ij}(\psi_n) = \tfrac{1}{2}\tr(A_i A_j)$,
  \[
    \{\mathcal{J}_n(\psi) - \mathcal{I}_n(\psi_n)\}_{ij} = R^\tsp \Phi A_i \Phi A_j \Phi R - \tfrac{1}{2}\tr(\Phi A_i \Phi A_j) - \tfrac{1}{2}\tr(A_i A_j)
  \]
  and
  \[
    w^\tsp \{\mathcal{J}_n(\psi) - \mathcal{I}_n(\psi_n)\} w = R^\tsp \Phi \bar{A} \Phi \bar{A} \Phi R - \tfrac{1}{2}\tr(\Phi \bar{A} \Phi \bar{A}) - \tfrac{1}{2}\tr(\bar{A}^2),
  \]
  where $\bar{A} = A(w, \psi_n) = \sum_i w_i A_i$. This expression depends on $\psi$ only through $\Phi$ since the $A_i$
  are evaluated at $\psi_n$.

  Let $s = \mathcal{I}_n(\psi_n)^{1/2}(\psi - \psi_n)$, $t = (s, u)$, and $H_t = \Phi \bar{A} \Phi \bar{A} \Phi$, which is symmetric because $\Phi$ and $\bar{A}$ are. Then the right-hand side of the last display is
  \[
    \{R^\tsp H_t R - \tr(H_t)\} + \{\tr(H_t) - \tfrac{1}{2}\tr(\Phi \bar{A} \Phi \bar{A}) - \tfrac{1}{2}\tr(\bar{A}^2)\},
  \]
  with the first term in braces a zero-mean Gaussian quadratic form and the second deterministic.

  We begin with the deterministic term, $\tr(H_t) -
  \tfrac{1}{2}\tr(\Phi\bar{A}\Phi\bar{A}) - \tfrac{1}{2}\tr(\bar{A}^2)$.
  Substituting $E := \Phi - I_{n-p}$ and expanding, the terms without a factor $E$
  cancel, and each remaining term is the trace of a product of two factors of
  $\bar{A}$ and factors $E$ or $I_{n-p}$, which is bounded in absolute value
  by $\Vert E\Vert^{j}\Vert\bar{A}\Vert_F^2$, with $j \geq 1$ the number of
  factors $E$, because $\vert\tr(X\bar{A}Y\bar{A}Z)\vert \leq \Vert
  ZX\Vert\,\Vert Y\Vert\,\Vert\bar{A}\Vert_F^2$. The surviving term with the fewest factors of $E$ has $j=1$, so it suffices to show $\| \bar{A} \|_F = O(1)$ and $\| E \| = O(M_n \bar{a_n})$, which gives that the deterministic term is $O(\| E \| \| \bar{A} \|^2_F) = O(M_n\bar{a}_n)$. For the former, $\Vert \bar{A}\Vert_F^2 = 2 w^\tsp \mathcal{I}_n(\psi_n) w = 2\Vert u\Vert^2 = 2$. For the latter, the identity $v^\tsp\mathcal{I}_n(\psi_n)v = \tr\{A(v, \psi_n)^2\}/2$ from Section~\ref{sec:gen_theory} gives $\Vert A(h, \psi_n)\Vert_F^2 = 2h^\tsp\mathcal{I}_n(\psi_n)h \leq 2M_n^2$. Combining with the same scale-invariance argument in the proof of Lemma \ref{lem:interior} yields $\Vert A(h, \psi_n)\Vert \leq \bar{a}_n \Vert A(h,
  \psi_n)\Vert_F \leq \sqrt{2}M_n\bar{a}_n$. Thus, the eigenvalues of $A(h,
  \psi_n)$ have magnitude less than one for all large enough $n$, and therefore
  we can write $E = \Phi - I_{n-p} = - \{I_{n-p} + A(h, \psi_n)\}^{-1} A(h,
  \psi_n)$. Then, if $\lambda_1, \dots, \lambda_{n - p}$ are the eigenvalues of
  $A(h, \psi_n)$, the symmetric matrix $E$ has eigenvalues $-\lambda_j/(1 + \lambda_j)$, so
  $\Vert E\Vert = \max_j \vert \lambda_j\vert/\vert 1 + \lambda_j\vert \leq \Vert A(h, \psi_n)\Vert/\{1 - \Vert A(h, \psi_n)\Vert\}
  \leq \sqrt{2}M_n\bar{a}_n/(1 - \sqrt{2}M_n\bar{a}_n) = O(M_n\bar{a}_n)$, using $\vert 1 + \lambda_j\vert \geq 1 - \vert\lambda_j\vert
  \geq 1 - \Vert A(h, \psi_n)\Vert$ and $\vert\lambda_j\vert \leq \Vert A(h, \psi_n)\Vert$. 

  Now consider the Gaussian quadratic form: we bound
  \[
    \sup_{t \in T} \left\vert R^\tsp H_t R - \tr H_t \right\vert, \qquad
    T = \{(s, u) : \Vert s\Vert \leq M_n,\ \Vert u\Vert = 1\},
  \]
  by verifying the conditions of Lemma~\ref{lem:qf_sup}. The index set $T$ lies
  in a ball of radius $\rho = (M_n^2 + 1)^{1/2}$ in $\R{2r}$, and $\rho =
  O(M_n)$ since $\liminf_n M_n > 0$. Because $\Vert E\Vert = O(M_n\bar{a}_n)
  \to 0$ uniformly over $T$, $\Vert \Phi\Vert \leq 1 + \Vert E\Vert \leq 2$ for
  all large $n$, and hence, using $\Vert \bar{A}\Vert \leq
  \bar{a}_n\Vert\bar{A}\Vert_F = \sqrt{2}\bar{a}_n$,
  \[
    \Vert H_t\Vert_F \leq
    \Vert \Phi\Vert^3 \Vert \bar{A}\Vert \Vert \bar{A}\Vert_F \leq 16\bar{a}_n,
    \qquad t \in T,
  \]
  so we may take $\sigma = 16\bar{a}_n = O(\bar{a}_n)$. Lastly, we verify that
  \[
    \Vert H_t - H_{t'}\Vert_F \leq L \Vert t - t'\Vert, \qquad t, t' \in T,
  \]
  for some $L$ not depending on $t, t'$, where we write $t' = (s', u')$. Because $\bar{A} = A(w, \psi_n)$ is linear in $u$ through $w =
  \mathcal{I}_n(\psi_n)^{-1/2}u$,
  \[
    \Vert \bar{A}_t - \bar{A}_{t'}\Vert_F = \sqrt{2}\Vert u - u'\Vert, \qquad
    \Vert \bar{A}_t - \bar{A}_{t'}\Vert \leq \sqrt{2}\bar{a}_n\Vert u - u'\Vert.
  \]
  Furthermore, $\Phi = \{I_{n-p} + A(\mathcal{I}_n(\psi_n)^{-1/2}s, \psi_n)\}^{-1}$ satisfies $\partial_s\Phi =
  -\Phi(\partial_s A)\Phi$ with $\Vert \partial_s A\Vert = O(\bar{a}_n)$;
  integrating along the segment from $s'$ to $s$, which lies in the ball $\Vert
  s\Vert \leq M_n$ by convexity, gives
  $\Vert \Phi_t - \Phi_{t'}\Vert = O(\bar{a}_n)\Vert s - s'\Vert$. Writing $H_t - H_{t'}$ as a sum of five terms, each the product
  $\Phi\bar{A}\Phi\bar{A}\Phi$ with one factor replaced by $\Phi_t - \Phi_{t'}$
  or $\bar{A}_t - \bar{A}_{t'}$ and the others taken from $t$ or $t'$, the
  bounds above give $\Vert H_t - H_{t'}\Vert_F = O(\bar{a}_n^2)\Vert s -
  s'\Vert + O(\bar{a}_n)\Vert u - u'\Vert = O(\bar{a}_n)\Vert t - t'\Vert$, so
  we may take $L = O(\bar{a}_n)$. Lemma~\ref{lem:qf_sup} with $\sigma = O(\bar{a}_n)$, $L = O(\bar{a}_n)$, and $\rho = O(M_n)$ then gives
  \[
  \E \sup_{t \in T} \vert R^\tsp H_t R - \tr H_t \vert \leq C_{2r} (\sigma + L\rho)
  \]
  with $C_{2r}$ depending only on $r$. Since $r$ is fixed and $\liminf_n M_n > 0$, $\sup_{t \in T} \vert R^\tsp H_t R - \tr H_t\vert = O_p(\bar{a}_n + M_n\bar{a}_n) = O_p(M_n\bar{a}_n)$. With the
  deterministic term, which is $O(M_n\bar{a}_n)$, the supremum over $t = (\mathcal{I}_n(\psi_n)^{1/2}(\psi - \psi_n), u)$, equivalent to the supremum over $(\psi, u)$, is $O_p(M_n\bar{a}_n)$. This proves the lemma.
\end{proof}

The next lemma is the analogue of Lemma~\ref{lem:local_ulln} for the expected
information, used to verify \eqref{eq:suff_info} below.

\begin{lemma} \label{lem:local_detrmn}
If \eqref{eq:negligible} holds, then for any radius $M_n$ with $M_n\bar{a}_n \to 0$,
\[
\sup_{\psi \in D_n^{M_n}} \Vert\mathcal{I}_n(\psi_n)^{-1/2} \mathcal{I}_n(\psi) \mathcal{I}_n(\psi_n)^{-1/2} - I_r\Vert = O(M_n\bar{a}_n).
\]
\end{lemma}

\begin{proof}
Consider $\Phi$, $\bar{A}$ and $E$ as defined in the proof of Lemma \ref{lem:local_ulln}. Pick arbitrary $u \in \mathbb{S}^{r - 1}$ and $\psi \in D_n^{M_n}$. As in Lemma \ref{lem:local_ulln}, direct calculation shows that the expected information evaluated at arbitrary $\psi$ has entries
\[
\mathcal{I}_{n, ij}(\psi) = \frac{1}{2}\tr(\Phi A_i \Phi A_j).
\]
Thus,
\[
\begin{aligned}
u^\tsp\left( \mathcal{I}_n(\psi_n)^{-1/2} \mathcal{I}_n(\psi) \mathcal{I}_n(\psi_n)^{-1/2} - I_r \right) u &= \frac{1}{2} \tr(\Phi \bar{A} \Phi \bar{A}) - \frac{1}{2} \tr(\bar{A}^2) \\
&= \tr(E\bar{A}^2) + \frac{1}{2}\tr(E \bar{A} E \bar{A})
\end{aligned}
\]
and
\[
\left\vert \tr(E\bar{A}^2) + \frac{1}{2} \tr(E\bar{A}E\bar{A}) \right\vert \leq \Vert E \Vert \Vert \bar{A} \Vert^2_F + \frac{1}{2} \Vert E \Vert^2 \Vert \bar{A} \Vert_F^2 = O(M_n\bar{a}_n)
\]
since, as in the proof of Lemma \ref{lem:local_ulln}, $\Vert \bar{A}\Vert_F = O(1)$ and $\Vert E\Vert = O(M_n\bar{a}_n)$. Taking the supremum over $u \in \mathbb{S}^{r-1}$ and $\psi \in D_n^{M_n}$ gives the result.
\end{proof}

We next record the norm of the matrix that multiplies the nuisance-score
residual when the root is inexact; it is the quantity referred to in
Section~\ref{sec:gen_theory}.

\begin{lemma} \label{lem:kappa}
  If $\mathcal{I}_n(\psi_n) \succ 0$, then $\kappa_n < 1$ and, with all
  matrices evaluated at $\psi_n$,
  \[
    \left\Vert \{\mathcal{I}_n^{(11)}\}^{1/2}\mathcal{I}_{n(12)}
      \mathcal{I}_{n(22)}^{-1/2}\right\Vert = \frac{\kappa_n}{\sqrt{1 - \kappa_n^2}}.
  \]
\end{lemma}

\begin{proof}
  Since $\mathcal{I}_n(\psi_n) \succ 0$, its Schur complement
  $(\mathcal{I}_n^{(11)})^{-1} = \mathcal{I}_{n(11)} -
  \mathcal{I}_{n(12)}\mathcal{I}_{n(22)}^{-1}\mathcal{I}_{n(21)} =
  \mathcal{I}_{n(11)}(1 - \kappa_n^2)$ is positive, where the second equality
  is \eqref{eq:kappa}, and hence
  $\kappa_n < 1$. The squared norm in the statement is
  $\mathcal{I}_n^{(11)}\mathcal{I}_{n(12)}\mathcal{I}_{n(22)}^{-1}
  \mathcal{I}_{n(21)} = \mathcal{I}_n^{(11)}\mathcal{I}_{n(11)}\kappa_n^2 =
  \kappa_n^2/(1 - \kappa_n^2)$, which completes the proof.
\end{proof}

By Lemma~\ref{lem:interior}, $(\psi_{n1}, \psi_{(2)}) \in \tilde{\Upsilon}_n$
for every $\psi_{(2)} \in B_n$ once $n$ is large, so the maximizer of
$\ell_n(\psi_{n1}, \cdot)$ over $B_n$, denoted $\hat{\psi}_{n(2)}$, is well
defined.

\begin{lemma} \label{lem:nuis_est}
  If \eqref{eq:negligible} holds, then, with probability tending to one,
  $\ell_n(\psi_{n1}, \cdot)$ is strictly concave on $B_n$, and
  $\hat{\psi}_{n(2)}$ is interior to $B_n$, is the unique root of
  $S_{n(2)}(\psi_{n1}, \cdot)$ there, and satisfies \eqref{eq:nuisance_root}
  and
  \[
    \mathcal{I}_{n(22)}(\psi_n)^{1/2}(\hat{\psi}_{n(2)} - \psi_{n(2)}) =
    \mathcal{I}_{n(22)}(\psi_n)^{-1/2} S_{n(2)}(\psi_n) + o_p(1),
  \]
  and hence \eqref{eq:nuis_rate}. If in addition $\tilde{\psi}_{n(2)}$
  satisfies \eqref{eq:nuisance_root} and \eqref{eq:nuis_rate}, then
  $\tilde{\psi}_{n(2)} = \hat{\psi}_{n(2)}$ with probability tending to one.
\end{lemma}

\begin{proof}
  Throughout, write $\mathcal{I}_{n(22)} = \mathcal{I}_{n(22)}(\psi_n)$ and,
  for $\psi_{(2)} \in B_n$, $\psi = (\psi_{n1}, \psi_{(2)}^\tsp)^\tsp$, $h =
  \psi - \psi_n$, and $w = \mathcal{I}_{n(22)}^{1/2}(\psi_{(2)} -
  \psi_{n(2)})$, such that $\Vert w\Vert \leq M_n$. Let $W_{n(2)} =
  \mathcal{I}_{n(22)}^{-1/2}S_{n(2)}(\psi_n)$, which is $O_p(1)$ because it
  has mean zero and $\var(W_{n(2)}) = I_{r-1}$ with $r$ fixed. For all large
  $n$, $D_n \subseteq \tilde{\Upsilon}_n$ by Lemma~\ref{lem:interior}, so
  $\ell_n(\psi_{n1}, \cdot)$ is twice continuously differentiable on the
  compact convex set $B_n$ and has a maximizer there; let $\hat{\psi}_{n(2)}$
  be one.

  We first show that, with probability tending to one, $\ell_n(\psi_{n1},
  \cdot)$ is strictly concave on $B_n$. Let $Q =
  \mathcal{I}_{n(22)}^{-1/2}P_2\mathcal{I}_n(\psi_n)^{1/2}$, which satisfies
  $\Vert Q\Vert = 1$ because $QQ^\tsp = I_{r-1}$. For every $\psi \in D_n$,
  \[
    \mathcal{I}_{n(22)}^{-1/2}\{\mathcal{J}_{n(22)}(\psi) -
    \mathcal{I}_{n(22)}\}\mathcal{I}_{n(22)}^{-1/2} = Q\left[
    \mathcal{I}_n(\psi_n)^{-1/2}\{\mathcal{J}_n(\psi) -
    \mathcal{I}_n(\psi_n)\}\mathcal{I}_n(\psi_n)^{-1/2}\right]Q^\tsp,
  \]
  and $(\psi_{n1}, \psi_{(2)}) \in D_n$ for every $\psi_{(2)} \in B_n$, so
  Lemma~\ref{lem:local_ulln} with radius $M_n$ gives
  \begin{equation} \label{eq:nuis_ulln}
    \sup_{\psi_{(2)} \in B_n}\left\Vert \mathcal{I}_{n(22)}^{-1/2}
      \{\mathcal{J}_{n(22)}(\psi_{n1}, \psi_{(2)}) - \mathcal{I}_{n(22)}\}
      \mathcal{I}_{n(22)}^{-1/2}\right\Vert = O_p(M_n\bar{a}_n).
  \end{equation}
  Because $M_n\bar{a}_n \to 0$, with probability tending to one
  $\mathcal{J}_{n(22)}(\psi_{n1}, \psi_{(2)}) \succ 0$ for every $\psi_{(2)}
  \in B_n$, and on that event $\ell_n(\psi_{n1}, \cdot)$ is strictly concave
  on $B_n$, so $\hat{\psi}_{n(2)}$ is its unique maximizer and its only
  stationary point in $B_n$.

  Next, for $\psi_{(2)} \in B_n$, the segment from $\psi_n$ to $\psi$ lies in
  the convex set $D_n$, so Taylor's theorem with integral remainder gives
  \[
    \ell_n(\psi) - \ell_n(\psi_n) = h^\tsp S_n(\psi_n) - \frac{1}{2}h^\tsp
      \left\{2\int_0^1(1 - t)\,\mathcal{J}_n(\psi_n + th)\,\dd t\right\}h.
  \]
  Because $h$ has zero interest block, $h^\tsp S_n(\psi_n) = (\psi_{(2)} -
  \psi_{n(2)})^\tsp S_{n(2)}(\psi_n) = w^\tsp W_{n(2)}$, and the quadratic
  term equals $w^\tsp\{I_{r-1} + \Delta_n(\psi_{(2)})\}w$, where
  \[
    \Delta_n(\psi_{(2)}) = 2\int_0^1(1 - t)\,\mathcal{I}_{n(22)}^{-1/2}
      \{\mathcal{J}_{n(22)}(\psi_n + th) - \mathcal{I}_{n(22)}\}
      \mathcal{I}_{n(22)}^{-1/2}\,\dd t.
  \]
  The nuisance block of $\psi_n + th$ lies in $B_n$ for $t \in [0, 1]$, so
  $\sup_{\psi_{(2)} \in B_n}\Vert\Delta_n(\psi_{(2)})\Vert = O_p(M_n\bar{a}_n)$
  by \eqref{eq:nuis_ulln} and $\int_0^1 2(1-t)dt  = 1$. Hence, 
  %uniformly in $\psi_{(2)} \in B_n$,
  \[
    2\{\ell_n(\psi) - \ell_n(\psi_n)\} = 2w^\tsp W_{n(2)} - \Vert w\Vert^2 +
      \tau_n(\psi_{(2)}) = -\Vert w - W_{n(2)}\Vert^2 + \Vert W_{n(2)}\Vert^2
      + \tau_n(\psi_{(2)}),
  \]
  where $\tau_n(\psi_{(2)}) = -w^\tsp\Delta_n(\psi_{(2)})w$ satisfies
  \[
    \sup_{\psi_{(2)} \in B_n}\vert\tau_n(\psi_{(2)})\vert \leq M_n^2
    \sup_{\psi_{(2)} \in B_n}\Vert\Delta_n(\psi_{(2)})\Vert = O_p(M_n^3\bar{a}_n)
    = o_p(1).
  \]

  Let $\hat{w}_n = \mathcal{I}_{n(22)}^{1/2}(\hat{\psi}_{n(2)} -
  \psi_{n(2)})$. Since $\Vert W_{n(2)}\Vert = O_p(1)$ and $M_n \to \infty$,
  the point $\psi^{W}_{(2)} = \psi_{n(2)} + \mathcal{I}_{n(22)}^{-1/2}W_{n(2)}$,
  whose whitened coordinate is $W_{n(2)}$, lies in $B_n$ with probability
  tending to one. On that event, $\ell_n(\psi_{n1}, \hat{\psi}_{n(2)}) \geq
  \ell_n(\psi_{n1}, \psi^{W}_{(2)})$, which by the last display gives
  \[
    \Vert\hat{w}_n - W_{n(2)}\Vert^2 \leq \tau_n(\hat{\psi}_{n(2)}) -
      \tau_n(\psi^{W}_{(2)}) \leq 2\sup_{\psi_{(2)} \in B_n}
      \vert\tau_n(\psi_{(2)})\vert = o_p(1).
  \]
  That is, $\hat{w}_n = W_{n(2)} + o_p(1)$, which is the expansion displayed
  in the statement. Consequently $\Vert\hat{w}_n\Vert = O_p(1)$, and since
  $M_n \to \infty$, $\Vert\hat{w}_n\Vert < M_n$ with probability tending to
  one, so $\hat{\psi}_{n(2)}$ is interior to $B_n$. Being an interior
  maximizer of the differentiable $\ell_n(\psi_{n1}, \cdot)$, it satisfies
  $S_{n(2)}(\psi_{n1}, \hat{\psi}_{n(2)}) = 0$, which is
  \eqref{eq:nuisance_root}, and it is the unique root of $S_{n(2)}(\psi_{n1},
  \cdot)$ in $B_n$ because it is the only stationary point there.

  For the final claim, suppose $\tilde{\psi}_{n(2)}$ satisfies
  \eqref{eq:nuisance_root} and \eqref{eq:nuis_rate}, and write 
  $$
  \tilde{w}_n = \mathcal{I}_{n(22)}(\psi_n)^{1/2}(\tilde{\psi}_{n(2)} - \psi_{n(2)}),
  $$
  so that $\tilde{\psi}_{n(2)} \in B_n$ if and only if $\Vert \tilde{w}_n\Vert
  \leq M_n$. Fix an arbitrary $\epsilon > 0$ and, using \eqref{eq:nuis_rate},
  pick $K$ with $\limsup_n \pr_{\psi_n}(\Vert \tilde{w}_n \Vert > K) <
  \epsilon$. Since $M_n \to \infty$, $M_n \geq K$ for all large $n$, so
  $\limsup_n \pr_{\psi_n}(\Vert \tilde{w}_n\Vert > M_n) < \epsilon$; as
  $\epsilon$ was arbitrary, $\tilde{\psi}_{n(2)} \in B_n$ with probability
  tending to one. Thus the event that $\hat{\psi}_{n(2)}$ is the unique root
  of $S_{n(2)}(\psi_{n1}, \cdot)$ in $B_n$, that $\tilde{\psi}_{n(2)} \in
  B_n$, and that $S_{n(2)}(\psi_{n1}, \tilde{\psi}_{n(2)}) = 0$ has
  probability tending to one, and on that event $\tilde{\psi}_{n(2)} =
  \hat{\psi}_{n(2)}$. This completes the proof.
\end{proof}

Under \eqref{eq:negligible}, $D_n \subseteq \tilde{\Upsilon}_n$ for all large
$n$ by Lemma~\ref{lem:interior}, so the maximizer of $\ell_n$ over $D_n$,
denoted $\bar{\psi}_n$, is well defined; write $\bar{\psi}_{n(2)}$ for its
nuisance block.

\begin{lemma} \label{lem:reml_start}
  If \eqref{eq:negligible} holds, then, with probability tending to one,
  $\ell_n$ is strictly concave on $D_n$, and $\bar{\psi}_n$ is interior to
  $D_n$, is the unique root of $S_n$ there, and satisfies
  \[
    \mathcal{I}_n(\psi_n)^{1/2}(\bar{\psi}_n - \psi_n) =
    \mathcal{I}_n(\psi_n)^{-1/2} S_n(\psi_n) + o_p(1);
  \]
  if in addition $\limsup_{n \to \infty} \kappa_n < 1$, then
  $\bar{\psi}_{n(2)}$ satisfies \eqref{eq:nuis_rate}.
\end{lemma}

\begin{proof}
  The proof of Lemma~\ref{lem:nuis_est} applies with the full parameter vector
  in place of the nuisance block, giving every claim up to and including the
  displayed expansion: replace $B_n$ by $D_n$, $h$ by $\psi - \psi_n$, $w$ by
  $\mathcal{I}_n(\psi_n)^{1/2}(\psi - \psi_n)$, and $W_{n(2)}$ by $W_n =
  \mathcal{I}_n(\psi_n)^{-1/2}S_n(\psi_n)$, which is $O_p(1)$ because it has
  mean zero and identity covariance with $r$ fixed. The identity $h^\tsp
  S_n(\psi_n) = w^\tsp W_n$ then holds by the definitions of $w$ and $W_n$,
  and the uniform bound on the observed information is
  Lemma~\ref{lem:local_ulln} on $D_n$ itself, without the projection by $Q$.

  For the final claim, write $\mathcal{I}_n = \mathcal{I}_n(\psi_n)$ and $\mathcal{I}_{n(22)} = \mathcal{I}_{n(22)}(\psi_n)$. Under \eqref{eq:negligible},
  $\mathcal{I}_n \succ 0$ for all large $n$. Since $\bar{\psi}_{n(2)} -
  \psi_{n(2)} = P_2(\bar{\psi}_n - \psi_n)$,
  \[
    \Vert \mathcal{I}_{n(22)}^{1/2}(\bar{\psi}_{n(2)} - \psi_{n(2)})\Vert \leq
    \Vert \mathcal{I}_{n(22)}^{1/2}P_2\mathcal{I}_n^{-1/2}\Vert\,
    \Vert \mathcal{I}_n^{1/2}(\bar{\psi}_n - \psi_n)\Vert,
  \]
  where the second factor is $O_p(1)$ by the expansion just established. For the first factor,
  \[
    \Vert \mathcal{I}_{n(22)}^{1/2}P_2\mathcal{I}_n^{-1/2}\Vert^2 = \Vert
    \mathcal{I}_{n(22)}^{1/2}P_2\mathcal{I}_n^{-1}P_2^\tsp\mathcal{I}_{n(22)}^{1/2}\Vert,
  \]
  and $P_2\mathcal{I}_n^{-1}P_2^\tsp = (\mathcal{I}_{n(22)} -
  \mathcal{I}_{n(21)}\mathcal{I}_{n(11)}^{-1}\mathcal{I}_{n(12)})^{-1}$ by the
  block form of the inverse, so, with $N = \mathcal{I}_{n(11)}^{-1/2}\mathcal{I}_{n(12)}\mathcal{I}_{n(22)}^{-1/2} \in \R{1 \times (r - 1)}$,
  \[
    \mathcal{I}_{n(22)}^{1/2}P_2\mathcal{I}_n^{-1}P_2^\tsp\mathcal{I}_{n(22)}^{1/2}
    = (I_{r-1} - N^\tsp N)^{-1}.
  \]
  $N^\tsp N$ has rank at most $1$ and is symmetric, and thus has one eigenvalue equal to $\Vert N\Vert^2 = \kappa_n^2 < 1$, with the rest $0$, by the definition of $\kappa_n$ and Lemma~\ref{lem:kappa}. Thus, the spectral norm of the right-hand side is $(1 - \kappa_n^2)^{-1}$. Hence
  \[
    \Vert \mathcal{I}_{n(22)}^{1/2}(\bar{\psi}_{n(2)} - \psi_{n(2)})\Vert =
    O_p\{(1 - \kappa_n^2)^{-1/2}\},
  \]
  which is $O_p(1)$ when $\limsup_{n\to\infty} \kappa_n < 1$, so
  $\bar{\psi}_{n(2)}$ satisfies \eqref{eq:nuis_rate}. This completes the
  proof.
\end{proof}

The efficient score for $\psi_1$ is
\begin{equation} \label{eq:eff_score}
  U_n(\psi) = S_{n,1}(\psi) - \mathcal{I}_{n(12)}(\psi_n)\mathcal{I}_{n(22)}(\psi_n)^{-1}S_{n(2)}(\psi),
\end{equation}
where the information matrices are evaluated at the true $\psi_n$ for every
argument $\psi$. Let $K_n =
\mathcal{I}_{n(12)}(\psi_n)\mathcal{I}_{n(22)}(\psi_n)^{-1}$, such that $U_n(\psi) = [1, -K_n]S_n(\psi)$,
and let $\tilde{\psi}_n = [\psi_{n1}, \tilde{\psi}_{n(2)}^\tsp]^\tsp$, where
$\tilde{\psi}_{n(2)} = \tilde{\psi}_{n(2)}(\psi_{n1})$ is the nuisance
estimator used to form $C_n(\alpha)$, evaluated at the true interest
parameter. By the theory for general parametric models in
\citet{shedden2026profile}, \eqref{eq:uniform_coverage} holds if, along every
sequence $(\psi_n)$ in $A$, $\mathcal{I}_n(\psi_n)$ is invertible for all
large $n$ and, with information matrices without an argument evaluated at
$\psi_n$,
\begin{align}
  \mathcal{I}_n^{-1/2}S_n(\psi_n) &\rightsquigarrow \rN(0, I_r),
    \label{eq:suff_clt}\\
  \{\mathcal{I}_n^{(11)}\}^{1/2}K_n S_{n(2)}(\tilde{\psi}_n) &= o_p(1),
    \label{eq:suff_nuis}\\
  \{\mathcal{I}_n^{(11)}\}^{1/2}\{U_n(\tilde{\psi}_n) - U_n(\psi_n)\} &= o_p(1),
    \label{eq:suff_linear}\\
  \mathcal{I}_n^{(11)}(\tilde{\psi}_n)/\mathcal{I}_n^{(11)} &= 1 + o_p(1).
    \label{eq:suff_info}
\end{align}

We now prove the results of Section~\ref{sec:gen_theory}.

\begin{proof}[Proof of Lemma~\ref{lem:sequences}]
  If $A$ is empty, then \eqref{eq:uniform_coverage} holds trivially, so
  suppose it is not, and write $\delta_n = \sup_{\psi \in A}\vert
  \pr_\psi\{\psi_1 \in C_n(\alpha)\} - (1 - \alpha)\vert$. Because $\delta_n$
  is finite, there is for each $n$ a $\psi_n \in A$ with
  $\vert \pr_{\psi_n}\{\psi_{n1} \in C_n(\alpha)\} - (1 - \alpha)\vert \geq
  \delta_n - 1/n$. The hypothesis applied to the sequence $(\psi_n)$ gives
  $\pr_{\psi_n}\{\psi_{n1} \in C_n(\alpha)\} \to 1 - \alpha$, and hence
  $\delta_n \leq \vert \pr_{\psi_n}\{\psi_{n1} \in C_n(\alpha)\} - (1 -
  \alpha)\vert + 1/n \to 0$. This completes the proof.
\end{proof}

\begin{proof}[Proof of Theorem~\ref{thm:porting}]
  Fix an arbitrary sequence $(\psi_n)$ in $A$. By Lemma~\ref{lem:nuis_est}, $\hat{\psi}_{n(2)}$
  is a root of $S_{n(2)}(\psi_{n1}, \cdot)$ with probability tending to one,
  satisfying \eqref{eq:nuisance_root} and \eqref{eq:nuis_rate}, and any $\tilde{\psi}_{n(2)}$ satisfying \eqref{eq:nuisance_root} and
  \eqref{eq:nuis_rate} is equal to it with probability tending to one. Fix such a $\tilde{\psi}_{n(2)}$. To show \eqref{eq:uniform_coverage}, it suffices to consider an arbitrary sequence $(\psi_n)$ in $A$ and show \eqref{eq:suff_clt}--\eqref{eq:suff_info} and that $\mathcal{I}_n(\psi_n)$ is invertible for large $n$. The latter follows from \eqref{eq:negligible} as noted at the beginning of this appendix. Throughout, information matrices without an argument are evaluated at $\psi_n$.
  
  For \eqref{eq:suff_clt}, fix $v$ with $\| v \| = 1$. The vector $V^\tsp Y$ follows
  the model of \citet{ekvall2026uniform} with known zero mean, $n - p$
  observations, and random-effect design $V^\tsp Z$, and their $a(v, \psi)$ is
  then our $a_n(v, \psi)$. By their
  Lemma~5, for every $n$ with $\bar{a}_n^2 < 1/8$, the difference between the density of
  $v^\tsp\mathcal{I}_n^{-1/2}S_n(\psi_n)$ and the standard normal density is bounded by
  $0.14\{4 + 0.29(1 - 8\bar{a}_n^2)^{-2}\}\bar{a}_n$, uniformly on $\R{}$.
  The proof of the lemma uses only that $v^\tsp\mathcal{I}_n^{-1/2}S_n(\psi_n) =
  [R^\tsp A(\mathcal{I}_n^{-1/2}v, \psi_n)R - \tr\{A(\mathcal{I}_n^{-1/2}v,
  \psi_n)\}]/2$, which holds for symmetric $H_j$. By
  \eqref{eq:negligible}, the density of $v^\tsp\mathcal{I}_n^{-1/2}S_n(\psi_n)$
  therefore converges pointwise to the standard normal density. Scheffé's theorem \citep[Theorem~16.12]{billingsley1995probability} then gives convergence in total variation, and thus in distribution. Since $r$ is fixed
  and $v$ was arbitrary, \eqref{eq:suff_clt} follows.

  Condition \eqref{eq:suff_nuis} holds because $\tilde{\psi}_{n(2)}$
  satisfies \eqref{eq:nuisance_root}.

  For \eqref{eq:suff_linear}, let $\gamma_n(t) = (\psi_{n1}, \psi_{n(2)}^\tsp +
  t(\tilde{\psi}_{n(2)} - \psi_{n(2)})^\tsp)^\tsp$ for $t \in [0, 1]$. Since
  $\gamma_n(t) - \psi_n$ has zero interest block, $\Vert
  \mathcal{I}_n^{1/2}\{\gamma_n(t) - \psi_n\}\Vert = t\Vert
  \mathcal{I}_{n(22)}^{1/2}(\tilde{\psi}_{n(2)} - \psi_{n(2)})\Vert$, which
  is $O_p(1)$ by \eqref{eq:nuis_rate}, so the chord $\{\gamma_n(t) : t \in
  [0, 1]\}$ lies in $D_n$ with probability tending to one because $M_n \to
  \infty$. As $U_n(\psi) = [1, -K_n]S_n(\psi)$, the map $t \mapsto
  U_n(\gamma_n(t))$ is continuously differentiable with derivative $-[1,
  -K_n]\mathcal{J}_n(\gamma_n(t))P_2^\tsp(\tilde{\psi}_{n(2)} -
  \psi_{n(2)})$, and $[1, -K_n]\mathcal{I}_nP_2^\tsp = \mathcal{I}_{n(12)} -
  K_n\mathcal{I}_{n(22)} = 0$, so
  \[
    U_n(\psi_{n1}, \tilde{\psi}_{n(2)}) - U_n(\psi_n) = -\int_0^1 [1, -K_n]
      \{\mathcal{J}_n(\gamma_n(t)) - \mathcal{I}_n\}\,\dd t\,
      P_2^\tsp(\tilde{\psi}_{n(2)} - \psi_{n(2)}).
  \]
  Let $G_n(\psi) := \mathcal{I}_n^{-1/2}\{\mathcal{J}_n(\psi) - \mathcal{I}_n\}\mathcal{I}_n^{-1/2}$ such that $\mathcal{J}_n(\psi) - \mathcal{I}_n =
  \mathcal{I}_n^{1/2}G_n(\psi)\mathcal{I}_n^{1/2}$ and
  \[
    \Xi_n = \{\mathcal{I}_n^{(11)}\}^{1/2}[1, -K_n]\mathcal{I}_n^{1/2} \in
    \R{1 \times r}.
  \]
  On the event that the chord lies in $D_n$,
  \begin{multline*}
    \left\Vert \{\mathcal{I}_n^{(11)}\}^{1/2}\{U_n(\psi_{n1},
      \tilde{\psi}_{n(2)}) - U_n(\psi_n)\}\right\Vert \\
    \leq \sup_{\psi \in D_n}\left\Vert \Xi_n G_n(\psi)\,\mathcal{I}_n^{1/2}
      P_2^\tsp\mathcal{I}_{n(22)}^{-1/2}\,\mathcal{I}_{n(22)}^{1/2}
      (\tilde{\psi}_{n(2)} - \psi_{n(2)})\right\Vert \\
    \leq \sup_{\psi \in D_n}\Vert G_n(\psi)\Vert\,
      \Vert \mathcal{I}_{n(22)}^{1/2}(\tilde{\psi}_{n(2)} - \psi_{n(2)})\Vert.
  \end{multline*}
  The second inequality uses two identities that hold exactly for every $n$,
  however close to singular the correlation matrix of the score is: $\Vert
  \Xi_n\Vert = 1$, because $\Xi_n\Xi_n^\tsp =
  \mathcal{I}_n^{(11)}(\mathcal{I}_{n(11)} -
  \mathcal{I}_{n(12)}\mathcal{I}_{n(22)}^{-1}\mathcal{I}_{n(21)}) = 1$, and
  $\Vert \mathcal{I}_n^{1/2}P_2^\tsp\mathcal{I}_{n(22)}^{-1/2}\Vert = 1$,
  because $\mathcal{I}_{n(22)}^{-1/2}P_2\mathcal{I}_nP_2^\tsp
  \mathcal{I}_{n(22)}^{-1/2} = I_{r-1}$. By Lemma~\ref{lem:local_ulln},
  $\sup_{\psi \in D_n}\Vert G_n(\psi)\Vert = O_p(M_n\bar{a}_n)$, and
  $\mathcal{I}_{n(22)}^{1/2}(\tilde{\psi}_{n(2)} - \psi_{n(2)}) = O_p(1)$ by
  \eqref{eq:nuis_rate}. Since $M_n\bar{a}_n \to 0$, \eqref{eq:suff_linear} follows.

  Finally, for \eqref{eq:suff_info}, let $\Omega_n(\psi) =
  \mathcal{I}_n^{-1/2}\mathcal{I}_n(\psi)\mathcal{I}_n^{-1/2}$ and $L = [1,
  0]\,\mathcal{I}_n^{-1/2} \in \R{1 \times r}$. Because
  $\mathcal{I}_n^{(11)}(\psi)$ is the leading element of
  $\mathcal{I}_n(\psi)^{-1}$, we have $\mathcal{I}_n^{(11)}(\psi_n) = LL^\tsp
  = \Vert L\Vert^2$ and $\mathcal{I}_n^{(11)}(\psi) =
  L\,\Omega_n(\psi)^{-1}L^\tsp$, so
  \[
    \left\vert \frac{\mathcal{I}_n^{(11)}(\psi)}{\mathcal{I}_n^{(11)}(\psi_n)}
      - 1\right\vert = \frac{\vert L\{\Omega_n(\psi)^{-1} -
      I_r\}L^\tsp\vert}{\Vert L\Vert^2} \leq \Vert \Omega_n(\psi)^{-1} -
      I_r\Vert.
  \]
  By Lemma~\ref{lem:local_detrmn}, $\sup_{\psi \in D_n}\Vert \Omega_n(\psi) -
  I_r\Vert = O(M_n\bar{a}_n)$, so
  \[
    \sup_{\psi \in D_n}\Vert \Omega_n(\psi)^{-1} - I_r\Vert \leq \sup_{\psi \in
    D_n}\frac{\Vert \Omega_n(\psi) - I_r\Vert}{1 - \Vert \Omega_n(\psi) -
    I_r\Vert} = O(M_n\bar{a}_n),
  \]
  and thus $\sup_{\psi \in D_n} \mathcal{I}_n^{(11)}(\psi)/\mathcal{I}_n^{(11)}(\psi_n) = 1 + O_p(M_n\bar{a}_n)$. By \eqref{eq:nuis_rate} and $M_n \to \infty$, the point $(\psi_{n1},
  \tilde{\psi}_{n(2)})$ lies in $D_n$ with probability tending to one, since
  its distance from $\psi_n$ in the $\mathcal{I}_n$ metric is $\Vert
  \mathcal{I}_{n(22)}^{1/2}(\tilde{\psi}_{n(2)} - \psi_{n(2)})\Vert = O_p(1)$.
  Hence $\mathcal{I}_n^{(11)}(\psi_{n1},
  \tilde{\psi}_{n(2)})/\mathcal{I}_n^{(11)}(\psi_n) = 1 + O_p(M_n\bar{a}_n) = 1
  + o_p(1)$, which is \eqref{eq:suff_info} and finishes the proof.
\end{proof}

\begin{proof}[Proof of Corollary~\ref{cor:maximizer}]
  Fix an arbitrary sequence $(\psi_n)$ in $A$. The set
  $\tilde{\Upsilon}^{(2)}_n(\psi_{n1})$ is open, so a maximizer of
  $\ell_n(\psi_{n1}, \cdot)$ over it is an interior maximizer of a
  differentiable function and hence satisfies \eqref{eq:nuisance_root}. It
  satisfies \eqref{eq:nuis_rate} by assumption, so
  Theorem~\ref{thm:porting} applies. This completes the proof.
\end{proof}

\begin{proof}[Proof of Proposition~\ref{prop:onestep}]
  Fix an arbitrary sequence $(\psi_n)$ in $A$. Write $\mathcal{I}_{n(22)}$
  for $\mathcal{I}_{n(22)}(\psi_n)$ and $\mathring{\psi}_{(2)}$ for
  $\mathring{\psi}_{n(2)}$. By \eqref{eq:negligible} and
  \eqref{eq:onestep_margin}, there is a
  sequence $M_n \to \infty$ with
  \[
    M_n\bar a_n \to 0, \qquad \frac{\kappa_n}{\sqrt{1 - \kappa_n^2}}\,
    M_n\bar a_n \to 0,
  \]
  for instance $M_n = d_n^{-1/2}$ with $d_n = \max\{\bar a_n, \kappa_n(1 -
  \kappa_n^2)^{-1/2}\bar a_n\}$. Let $\zeta_n$ be the supremum over
  $\psi_{(2)} \in B_n^{M_n}$ of
  \[
    \left\Vert \mathcal{I}_{n(22)}^{-1/2}\{\mathcal{J}_{n(22)}(\psi_{n1}, \psi_{(2)}) -
      \mathcal{I}_{n(22)}\}\mathcal{I}_{n(22)}^{-1/2}\right\Vert +
    \left\Vert \mathcal{I}_{n(22)}^{-1/2}\{\mathcal{I}_{n(22)}(\psi_{n1}, \psi_{(2)}) -
      \mathcal{I}_{n(22)}\}\mathcal{I}_{n(22)}^{-1/2}\right\Vert,
  \]
  which is $O_p(M_n\bar a_n)$ by Lemmas~\ref{lem:interior} to
  \ref{lem:local_detrmn}, each with radius $M_n$, and the projection by $Q$ in the proof of
  Lemma~\ref{lem:nuis_est}. Because $\mathring{\psi}_{(2)}$ satisfies
  \eqref{eq:nuis_rate} and $M_n \to \infty$, the point $(\psi_{n1},
  \mathring{\psi}_{(2)})$ lies in the convex set $D_n^{M_n}$ with probability
  tending to one, and we argue on that event. 
  Since $\partial
  S_{n(2)}/\partial\psi_{(2)}^\tsp = -\mathcal{J}_{n(22)}$, the fundamental
  theorem of calculus along the segment from $\psi_n$ to $(\psi_{n1},
  \mathring{\psi}_{(2)})$ gives
  \[
    S_{n(2)}(\psi_{n1}, \mathring{\psi}_{(2)}) - S_{n(2)}(\psi_n) =
    -\int_0^1 \mathcal{J}_{n(22)}\{\psi_{n1}, \psi_{n(2)} +
    t(\mathring{\psi}_{(2)} - \psi_{n(2)})\}\,\dd t\,(\mathring{\psi}_{(2)} -
    \psi_{n(2)}),
  \]
  and hence
  \[
    \mathcal{I}_{n(22)}^{-1/2}S_{n(2)}(\psi_{n1}, \mathring{\psi}_{(2)}) =
      \mathcal{I}_{n(22)}^{-1/2}S_{n(2)}(\psi_n) - (I_{r-1} + a_n)
      \mathcal{I}_{n(22)}^{1/2}(\mathring{\psi}_{(2)} - \psi_{n(2)}),
  \]
  where
  \[
    a_n = \int_0^1 \mathcal{I}_{n(22)}^{-1/2}\left[\mathcal{J}_{n(22)}\{\psi_{n1},
      \psi_{n(2)} + t(\mathring{\psi}_{(2)} - \psi_{n(2)})\} -
      \mathcal{I}_{n(22)}\right]\mathcal{I}_{n(22)}^{-1/2}\,\dd t.
  \]
  Since $\psi_{n(2)}, \mathring{\psi}_{(2)} \in B_n^{M_n}$, the latter on the
  event above, and $B_n^{M_n}$ is convex, the point $\psi_{n(2)} +
  t(\mathring{\psi}_{(2)} - \psi_{n(2)})$ lies in $B_n^{M_n}$, for every $t \in [0, 1]$. Thus, each integrand above has norm at
  most $\zeta_n$, and by $\int_0^1 \dd t = 1$,
  $\Vert a_n\Vert \leq \zeta_n = O_p(M_n\bar{a}_n)$. Thus the left-hand side
  is $O_p(1)$, the first term on the right having mean zero and identity
  covariance and the second being $O_p(1)$ by \eqref{eq:nuis_rate}.

  Define $b_n$ by
  \[
    \mathcal{I}_{n(22)}^{-1/2}\mathcal{I}_{n(22)}(\psi_{n1},
    \mathring{\psi}_{(2)})\mathcal{I}_{n(22)}^{-1/2} = I_{r-1} + b_n.
  \]
  Since $\mathring{\psi}_{(2)} \in B_n^{M_n}$, $b_n$ is less than or equal to the second term of $\zeta_n$, so $\Vert
  b_n\Vert \leq \zeta_n = O_p(M_n\bar a_n)$. In particular, $\Vert
  b_n\Vert < 1$ with probability tending to one, and on that event $I_{r-1} +
  b_n$ and $\mathcal{I}_{n(22)}(\psi_{n1}, \mathring{\psi}_{(2)})$ are
  invertible, $\mathcal{I}_{n(22)}^{1/2}\mathcal{I}_{n(22)}(\psi_{n1},
  \mathring{\psi}_{(2)})^{-1}\mathcal{I}_{n(22)}^{1/2} = (I_{r-1} +
  b_n)^{-1}$, and $\Vert (I_{r-1} + b_n)^{-1}\Vert \leq (1 - \Vert
  b_n\Vert)^{-1}$. By \eqref{eq:onestep},
  \[
    \mathcal{I}_{n(22)}^{1/2}(\check{\psi}_{n(2)} - \psi_{n(2)}) =
      \mathcal{I}_{n(22)}^{1/2}(\mathring{\psi}_{(2)} - \psi_{n(2)}) +
      (I_{r-1} + b_n)^{-1}\mathcal{I}_{n(22)}^{-1/2}S_{n(2)}(\psi_{n1},
      \mathring{\psi}_{(2)}),
  \]
  so $\check{\psi}_{n(2)}$ satisfies \eqref{eq:nuis_rate}, and
  $\check{\psi}_{n(2)} \in B_n^{M_n}$ with probability tending to one. We argue on
  that event as well. Rearranging the above gives
  $$
  \mathcal{I}_{n(22)}^{1/2}(\check{\psi}_{n(2)} -
  \mathring{\psi}_{(2)}) = (I_{r-1} +
  b_n)^{-1}\mathcal{I}_{n(22)}^{-1/2}S_{n(2)}(\psi_{n1},
  \mathring{\psi}_{(2)}).
  $$
  The fundamental theorem of calculus along the segment
  from $(\psi_{n1}, \mathring{\psi}_{(2)})$ to $(\psi_{n1},
  \check{\psi}_{n(2)})$ gives
  \[
    \mathcal{I}_{n(22)}^{-1/2}S_{n(2)}(\psi_{n1}, \check{\psi}_{n(2)}) =
    \mathcal{I}_{n(22)}^{-1/2}S_{n(2)}(\psi_{n1}, \mathring{\psi}_{(2)}) -
    (I_{r-1} + c_n)\mathcal{I}_{n(22)}^{1/2}(\check{\psi}_{n(2)} -
    \mathring{\psi}_{(2)}),
  \]
  where
  \[
    c_n = \int_0^1 \mathcal{I}_{n(22)}^{-1/2}\left[\mathcal{J}_{n(22)}\{\psi_{n1},
      \mathring{\psi}_{(2)} + t(\check{\psi}_{n(2)} - \mathring{\psi}_{(2)})\}
      - \mathcal{I}_{n(22)}\right]\mathcal{I}_{n(22)}^{-1/2}\,\dd t.
  \]
  Since $\mathring{\psi}_{(2)}, \check{\psi}_{n(2)} \in B_n^{M_n}$, both on
  events already fixed above, and $B_n^{M_n}$ is convex, the point
  $\mathring{\psi}_{(2)} + t(\check{\psi}_{n(2)} - \mathring{\psi}_{(2)})$
  lies in $B_n^{M_n}$, for every $t
  \in [0, 1]$; the same argument as for $a_n$ gives $\Vert c_n\Vert \leq
  \zeta_n = O_p(M_n\bar{a}_n)$. Combining the last two display equations then
  gives
  \begin{equation} \label{eq:onestep_resid}
    \mathcal{I}_{n(22)}^{-1/2}S_{n(2)}(\psi_{n1}, \check{\psi}_{n(2)}) =
      (b_n - c_n)(I_{r-1} + b_n)^{-1}\mathcal{I}_{n(22)}^{-1/2}
      S_{n(2)}(\psi_{n1}, \mathring{\psi}_{(2)}) = O_p(M_n\bar a_n),
  \end{equation}
  because $I_{r-1} - (I_{r-1} + c_n)(I_{r-1} + b_n)^{-1} = (b_n -
  c_n)(I_{r-1} + b_n)^{-1}$, $\Vert b_n - c_n\Vert \leq \Vert b_n\Vert +
  \Vert c_n\Vert \leq 2\zeta_n = O_p(M_n\bar a_n)$, $\Vert (I_{r-1} +
  b_n)^{-1}\Vert = O_p(1)$, and $\mathcal{I}_{n(22)}^{-1/2}S_{n(2)}(\psi_{n1},
  \mathring{\psi}_{(2)}) = O_p(1)$.

  With $\tilde{\psi}_{n(2)} = \check{\psi}_{n(2)}$, the invertibility of
  $\mathcal{I}_n(\psi_n)$, \eqref{eq:suff_clt}, \eqref{eq:suff_linear}, and
  \eqref{eq:suff_info} are verified as in
  the proof of Theorem~\ref{thm:porting}, which uses the estimator only
  through \eqref{eq:nuis_rate}. For \eqref{eq:suff_nuis}, by \eqref{eq:eff_score},
  Lemma~\ref{lem:kappa}, and \eqref{eq:onestep_resid},
  \begin{align*}
    \left\Vert \{\mathcal{I}_n^{(11)}\}^{1/2}\mathcal{I}_{n(12)}
      \mathcal{I}_{n(22)}^{-1}S_{n(2)}(\psi_{n1}, \check{\psi}_{n(2)})
      \right\Vert &\leq \frac{\kappa_n}{\sqrt{1 - \kappa_n^2}}\left\Vert
      \mathcal{I}_{n(22)}^{-1/2}S_{n(2)}(\psi_{n1}, \check{\psi}_{n(2)})
      \right\Vert \\
      &= O_p\left\{\frac{\kappa_n M_n\bar a_n}{\sqrt{1 - \kappa_n^2}}\right\},
  \end{align*}
  which is $o_p(1)$ by the choice of $M_n$. Because the sequence was
  arbitrary, \eqref{eq:uniform_coverage} follows. This completes
  the proof.
\end{proof}

\begin{proof}[Proof of Proposition~\ref{prop:common_design}]
  Let $H \in \R{m \times (m - 1)}$ satisfy $H^\tsp H = I_{m - 1}$ and $H^\tsp
  1_m = 0$, and let $V_1 \in \R{k \times (k - p)}$ satisfy $V_1^\tsp V_1 =
  I_{k - p}$ and $V_1^\tsp X_1 = 0$; if $k = p$, $V_1$ is empty and the
  second block in every block-diagonal matrix below is absent. The matrix $V
  = [H \otimes I_k,\ m^{-1/2}1_m \otimes V_1] \in \R{n \times (n - p)}$ is
  semi-orthogonal with $V^\tsp X = 0$, and we use it throughout; this is no
  restriction, because $\ell_n$ and $a_n$ do not depend on the choice of $V$.
  For $v \in \R{r}$, because $\Sigma(v) = I_m \otimes \Sigma_1(v)$ and
  $H^\tsp 1_m = 0$,
  \begin{equation} \label{eq:common_block}
    \tilde{\Sigma}(v) = V^\tsp\Sigma(v)V = \bdiag\{I_{m - 1} \otimes
      \Sigma_1(v),\ V_1^\tsp\Sigma_1(v)V_1\}.
  \end{equation}
  In particular, $\tilde{\Upsilon}_n = \{\psi : \Sigma_1(\psi) \succ 0\}$.

  Fix $\psi \in \tilde{\Upsilon}_n$ and a
  non-zero $v \in \R{r}$. Write $\Sigma_1 = \Sigma_1(\psi)$, $A_1 =
  \Sigma_1^{-1/2}\Sigma_1(v)\Sigma_1^{-1/2}$, and $Q = \Sigma_1^{1/2}V_1
  (V_1^\tsp\Sigma_1V_1)^{-1/2}$, which satisfies $Q^\tsp Q = I_{k - p}$. By
  \eqref{eq:common_block},
  \[
    A(v, \psi) = \bdiag(I_{m - 1} \otimes A_1,\ Q^\tsp A_1 Q).
  \]
  Since $\Vert Q^\tsp A_1 Q\Vert \leq \Vert A_1\Vert$, the spectral norm of
  $A(v, \psi)$ is $\Vert A_1\Vert$, and $\Vert A(v, \psi)\Vert_F^2 \geq (m -
  1)\Vert A_1\Vert_F^2$. Because $\Sigma_1(v) \neq 0$, $A_1 \neq 0$, and
  hence $a_n(v, \psi)^2 \leq \Vert A_1\Vert^2/\{(m - 1)\Vert A_1\Vert_F^2\}
  \leq (m - 1)^{-1}$. Taking the supremum over $v$ gives the bound on
  $\bar{a}_n$.

  Fix $\psi^0 \in \Upsilon$ as the
  true parameter. Partition $V^\tsp Y = [\vecop(Y_c)^\tsp, y_0^\tsp]^\tsp$
  with $Y_c \in \R{k \times (m - 1)}$ and $y_0 \in \R{k - p}$, and let $S =
  Y_cY_c^\tsp/(m - 1)$. By \eqref{eq:common_block}, the columns of $Y_c$ are
  independent $\rN_k\{0, \Sigma_1(\psi^0)\}$ and $y_0 \sim \rN_{k - p}\{0,
  V_1^\tsp\Sigma_1(\psi^0)V_1\}$ is independent of $Y_c$. Since $m - 1 \geq
  k$ and $Y_c$ has a density, $S \succ 0$ with probability one, and the rest
  of the proof is on this event. For $\psi \in \tilde{\Upsilon}_n$, let
  $\Sigma_e(\psi) = V_1^\tsp\Sigma_1(\psi)V_1$ and let $\nu_1(\psi) \leq
  \dots \leq \nu_k(\psi)$ be the eigenvalues of $\Sigma_1(\psi)^{-1}S$, which
  are positive. Let also $\phi(x) = x - \log x - 1$ for $x > 0$, which is
  non-negative, vanishes only at $x = 1$, tends to infinity as $x \to 0$ and
  as $x \to \infty$, and, by Taylor's theorem with $\phi(1) = \phi'(1) = 0$
  and $\phi''(x) = x^{-2}$, satisfies $(x - 1)^2/8 \leq \phi(x) \leq 2(x -
  1)^2$ for $1/2 \leq x \leq 2$. Using \eqref{eq:common_block} in
  \eqref{eq:res_ll}, and since $\log\vert\Sigma_1(\psi)\vert = \log\vert
  S\vert - \sum_{j = 1}^k\log\nu_j(\psi)$ and $\tr\{\Sigma_1(\psi)^{-1}S\} =
  \sum_{j = 1}^k\nu_j(\psi)$,
  \begin{equation} \label{eq:common_ll}
    \ell_n(\psi) = -\frac{m - 1}{2}\left[\log\vert S\vert + k + \sum_{j =
      1}^k\phi\{\nu_j(\psi)\}\right] - \frac{1}{2}\left\{
      \log\vert\Sigma_e(\psi)\vert + y_0^\tsp\Sigma_e(\psi)^{-1}y_0\right\}.
  \end{equation}
  Fix $\psi, \psi' \in \tilde{\Upsilon}_n$. Since $\nu_1(\psi')\Sigma_1(\psi')
  \preceq S \preceq \nu_k(\psi)\Sigma_1(\psi)$, also $\nu_1(\psi')
  \Sigma_e(\psi') \preceq V_1^\tsp S V_1 \preceq \nu_k(\psi)\Sigma_e(\psi)$,
  and hence $\log\vert\Sigma_e(\psi')\vert - \log\vert\Sigma_e(\psi)\vert
  \leq (k - p)\{\log\nu_k(\psi) - \log\nu_1(\psi')\}$. In addition, for $a,
  b > 0$, the convex function $x \mapsto a\phi(x) - b\log x = a(x - 1) - (a
  + b)\log x$ is minimized at $x = 1 + b/a$, where it equals $b - (a +
  b)\log(1 + b/a) \geq -b^2/a$; with $x = \nu_k(\psi)$, $a = (m - 1)/2$, $b
  = k - p$, and $(k - p)^2 \leq k(m - 1)$, this gives $(k - p)\log\nu_k(\psi)
  \leq (m - 1)\phi\{\nu_k(\psi)\}/2 + 2k$. Using these two bounds in the
  difference of \eqref{eq:common_ll} at $\psi$ and at $\psi'$, and dropping
  the non-negative quadratic form at $\psi$, gives
  \begin{equation} \label{eq:common_upper}
    \ell_n(\psi) - \ell_n(\psi') \leq k + \frac{\varrho(\psi')}{2} - \frac{m -
      1}{4}\sum_{j = 1}^k\phi\{\nu_j(\psi)\}, \qquad \psi, \psi' \in
      \tilde{\Upsilon}_n,
  \end{equation}
  where $\varrho(\psi') = (m - 1)\sum_{j = 1}^k\phi\{\nu_j(\psi')\} - (k -
  p)\log\nu_1(\psi') + y_0^\tsp\Sigma_e(\psi')^{-1}y_0$.

  Fix $\psi_1$ with
  $\tilde{\Upsilon}^{(2)}_n(\psi_1)$ non-empty and $\psi_{(2)}^0$ in that
  set, and let $\psi' = (\psi_1, \psi_{(2)}^{0\tsp})^\tsp$. By
  \eqref{eq:common_upper}, every $\psi_{(2)} \in
  \tilde{\Upsilon}^{(2)}_n(\psi_1)$ with $\ell_n(\psi_1, \psi_{(2)}) \geq
  \ell_n(\psi')$ satisfies $\sum_{j = 1}^k\phi\{\nu_j(\psi_1, \psi_{(2)})\}
  \leq C = \{4k + 2\varrho(\psi')\}/(m - 1)$. Since $\phi$ is continuous and
  tends to infinity at zero and at infinity, there is $\varepsilon \in (0,
  1]$ such that $\phi(x) \leq C$ implies $\varepsilon \leq x \leq
  1/\varepsilon$, and hence $\varepsilon S \preceq \Sigma_1(\psi_1,
  \psi_{(2)}) \preceq S/\varepsilon$ for every such $\psi_{(2)}$. The
  matrices between $\varepsilon S$ and $S/\varepsilon$ in the positive
  semi-definite order form a compact set of positive definite matrices, and
  $\psi_{(2)} \mapsto \Sigma_1(\psi_1, \psi_{(2)})$ is affine and injective
  because $\Sigma_1(v) \neq 0$ for $v \neq 0$. Thus, the superlevel set
  $\{\psi_{(2)} : \ell_n(\psi_1, \psi_{(2)}) \geq \ell_n(\psi')\}$ is
  relatively closed in a compact subset of $\tilde{\Upsilon}^{(2)}_n(\psi_1)$,
  so it is compact, and the continuous $\ell_n(\psi_1, \cdot)$ attains its
  maximum over $\tilde{\Upsilon}^{(2)}_n(\psi_1)$ on it.

  Fix a sequence $(\psi_n)$ in $\Upsilon$, take $\psi^0 =
  \psi_n$ above, write $\Sigma_n = \Sigma_1(\psi_n)$ and $\tilde{\psi}_n =
  (\psi_{n1}, \tilde{\psi}_{n(2)}^\tsp)^\tsp$, and let $\tilde{\nu}_j =
  \nu_j(\tilde{\psi}_n)$ and $w_j = \nu_j(\psi_n)$. The $w_j$ are the
  eigenvalues of $W = \Sigma_n^{-1/2}S\Sigma_n^{-1/2}$, and $(m - 1)W$ has
  the Wishart distribution with $m - 1$ degrees of freedom and identity
  scale, while $y_0^\tsp\Sigma_e(\psi_n)^{-1}y_0$ has the
  $\chi^2_{k - p}$-distribution; neither distribution depends on $\psi_n$, so
  every stochastic order bound below holds along every sequence in
  $\Upsilon$. Because $k$ is fixed, the entries of $W - I_k$ have mean zero
  and variance of order $m^{-1}$, so $\delta_n = \max_j\vert w_j - 1\vert =
  O_p(m^{-1/2})$. On the event $\delta_n \leq 1/2$, whose probability tends
  to one, $\phi(w_j) \leq 2\delta_n^2$ for every $j$ and $-\log w_1 \leq \log
  2$, so $\varrho(\psi_n) = O_p(1)$. Since $\tilde{\psi}_{n(2)}$ maximizes
  $\ell_n(\psi_{n1}, \cdot)$, $\ell_n(\tilde{\psi}_n) \geq \ell_n(\psi_n)$,
  and \eqref{eq:common_upper} with $\psi = \tilde{\psi}_n$ and $\psi' =
  \psi_n$ gives $\sum_{j = 1}^k\phi(\tilde{\nu}_j) \leq \{4k +
  2\varrho(\psi_n)\}/(m - 1) = O_p(m^{-1})$. For every $\eta > 0$, $\phi$ is
  bounded away from zero on $\{x : \vert x - 1\vert \geq \eta\}$, so
  $\tilde{\nu}_j \to 1$ in probability for every $j$. On the event
  $\max_j\vert\tilde{\nu}_j - 1\vert \leq 1/2$, whose probability therefore
  tends to one, $(\tilde{\nu}_j - 1)^2 \leq 8\phi(\tilde{\nu}_j)$, and hence
  $\varepsilon_n = \max_j\vert\tilde{\nu}_j - 1\vert = O_p(m^{-1/2})$.

  By the definitions of $\varepsilon_n$ and $\delta_n$, $(1 -
  \varepsilon_n)\Sigma_1(\tilde{\psi}_n) \preceq S \preceq (1 +
  \varepsilon_n)\Sigma_1(\tilde{\psi}_n)$ and $(1 - \delta_n)\Sigma_n \preceq
  S \preceq (1 + \delta_n)\Sigma_n$, so on the event $\varepsilon_n < 1$,
  whose probability tends to one, the eigenvalues of
  $\Sigma_n^{-1/2}\Sigma_1(\tilde{\psi}_n)\Sigma_n^{-1/2}$ lie between $(1 -
  \delta_n)/(1 + \varepsilon_n)$ and $(1 + \delta_n)/(1 - \varepsilon_n)$.
  Let $h = \tilde{\psi}_n - \psi_n$, whose first element is zero. By
  linearity, $\Sigma_1(\tilde{\psi}_n) - \Sigma_n = \Sigma_1(h)$, so $A_1 =
  \Sigma_n^{-1/2}\Sigma_1(h)\Sigma_n^{-1/2}$ is that matrix minus $I_k$, and
  $\Vert A_1\Vert \leq (\delta_n + \varepsilon_n)/(1 - \varepsilon_n) =
  O_p(m^{-1/2})$. By the block form of $A(v, \psi)$ above with $\psi =
  \psi_n$ and $v = h$, the identity $v^\tsp\mathcal{I}_n(\psi)v = \tr\{A(v,
  \psi)^2\}/2$ from Section~\ref{sec:gen_theory}, and since the squared
  Frobenius norm of a symmetric matrix is at most its dimension times its
  squared spectral norm,
  \[
    \Vert \mathcal{I}_{n(22)}(\psi_n)^{1/2}(\tilde{\psi}_{n(2)} -
      \psi_{n(2)})\Vert^2 = h^\tsp\mathcal{I}_n(\psi_n)h = \frac{1}{2}\Vert
      A(h, \psi_n)\Vert_F^2 \leq \frac{mk}{2}\Vert A_1\Vert^2 = O_p(1),
  \]
  where the first equality holds because the first element of $h$ is zero.
  This is \eqref{eq:nuis_rate}.

  Finally, along any sequence $(\psi_n)$ in $\Upsilon$, $\bar{a}_n(\psi_n)
  \leq (m - 1)^{-1/2} \to 0$, which is \eqref{eq:negligible}, and
  \eqref{eq:nuis_rate} was established above.
  Corollary~\ref{cor:maximizer} with $A = \Upsilon$ completes the proof.
\end{proof}

% app_computing.tex -- included via \input from 'compute lmm.tex'

\section{Computing details} \label{sec:lmm_computing}

The statistic of Section~\ref{sec:model} is built from the restricted
likelihood, but its parts are computed from the likelihood of
\eqref{eq:marginal}: the restricted score is the score for that likelihood
evaluated at $\tilde{\beta}(\psi)$, corrected by a determinant term, and the
restricted information has the same traces with a projection inserted. We
therefore give the likelihood, score, and information first, assemble the
restricted versions in Section~\ref{sec:reloglik}, and invert the statistic
last.

\subsection{Likelihood, score, and information} \label{sec:likelihood computation}

The log-likelihood for $(\beta, \psi)$ based on \eqref{eq:marginal} is, up to a
constant,
\begin{align} \label{eq:loglik}
 \bar{\ell}(\beta, \psi) =  -\frac{1}{2}\{\log \vert \Sigma(\psi)\vert + (Y - X\beta )^\tsp \Sigma(\psi)^{-1}(Y - X\beta)\},
\end{align}
where the bar distinguishes it from the restricted log-likelihood
\eqref{eq:res_ll} of the main text. Similarly, $\bar{S}$ and $\bar{\mathcal{I}}$
are the score and expected information based on $\bar{\ell}$ while $S$ and
$\mathcal{I}$ are the restricted ones of Section~\ref{sec:model}. We take the
parameter of interest to be a random effect parameter, but similar ideas apply
to a combination of random and fixed effect parameters.

Since $\bar{\ell}(\beta, \psi)$ is partially maximized in $\beta$ for fixed
$\psi$ by the $\tilde{\beta}(\psi)$ of Section~\ref{sec:model}, the test
statistic is based on the profile likelihood $\bar{\ell}\{\tilde{\beta}(\psi),
\psi\}$, and the nuisance estimator $\tilde{\psi}_{(2)}(\psi_1)$ is a local
maximizer of $\psi_{(2)} \mapsto \bar{\ell}\{\tilde{\beta}(\psi), \psi_1,
\psi_{(2)}\}$ over the $\psi_{(2)}$ for which $\Sigma(\psi)$ is positive
definite, the analogue for the full likelihood of the extended set
$\tilde{\Upsilon}^{(2)}(\psi_1)$ of Section~\ref{sec:model}.

The nonconvex objective can have flat regions and multiple local maxima. Trust-region
algorithms have proven useful in similar settings
\citep{knudson2021likelihoodbased}, and the expected information used in place
of the Hessian in Section~\ref{sec:computation} is positive definite, so the
model problem is well defined at every iterate. In practice, we have rarely
encountered convergence issues.

Evaluating the statistic requires the value and gradient of the profile
log-likelihood and the expected information. Since $\tilde{\beta}(\psi)$ is an
interior partial maximizer, the gradient of $\psi \mapsto
\bar{\ell}\{\tilde{\beta}(\psi),\psi\}$ is the score for $\psi$ evaluated at
$\beta = \tilde{\beta}(\psi)$.

We focus on the case where $p$ is small relative to $n$ and the random-effects design is sparse, which is common in practice. Then, a useful identity is
\begin{equation} \label{eq:sigma_woodbury}
  \Sigma(\psi)^{-1} = \psi_r^{-1} \{I_n - Z(I_q  + \Psi_r Z^\tsp Z)^{-1} \Psi_r Z^\tsp\} = \psi_r^{-1}\{I_n - Z A(\psi) Z^\tsp\},
\end{equation}
where $\Psi_r = \Psi / \psi_r$ and $A(\psi) = (I_q + \Psi_r Z^\tsp
Z)^{-1}\Psi_r$ \citep[see e.g.,][]{henderson1981deriving}. The $q\times q$ matrix $A(\psi)$
is symmetric and turns up in many computations to follow. Thus, it is computed
once and stored. The matrix $B(\psi) = I_q + \Psi_r Z^\tsp Z$ is typically
sparse, so we use a sparse solver to compute $A(\psi) =
B(\psi)^{-1}\Psi_r$, whose fill determines the cost of the multiplications
that follow. If $B(\psi)$ were dense,
the factorization would cost $O(q^3)$; with sparsity the cost is lower and
depends on the sparsity pattern, as do the costs of the matrix multiplications
involving $Z$ and $\Psi$.

To evaluate the log-likelihood in \eqref{eq:loglik} efficiently, note
\[
  \vert \Sigma(\psi)\vert = \psi_r^n \vert I_n + Z \Psi_r Z^\tsp\vert = \psi_r^n \vert I_q + \Psi_r Z^\tsp Z\vert = \psi_r^n \vert B(\psi)\vert,
\]
which requires $O(q)$ operations using the available factorization of $B(\psi)$. To simplify notation, we write the quadratic term of $\bar{\ell}(\beta, \psi)$ as
\begin{align*}
  (Y - X\beta)^\tsp \Sigma(\psi)^{-1}(Y - X\beta) &= e(\beta)^\tsp \tilde{e}(\beta, \psi),
\end{align*}
where $e(\beta) = Y - X\beta$ and $\tilde{e}(\beta, \psi) = \Sigma(\psi)^{-1}(Y
- X\beta)$. To compute the latter efficiently, note, by
\eqref{eq:sigma_woodbury}, $ \tilde{e}(\beta, \psi) = \psi_r^{-1}\{e(\beta) - Z
A(\psi) Z^\tsp e(\beta)\}$; without sparsity, the cost would be $O(nq + q^2)$.
Consider next the score functions for $\psi$,
\begin{align} \label{eq:score_psi_r}
 \bar{S}(\psi_r) = \frac{1}{2}[\tilde{e}(\beta, \psi)^\tsp \tilde{e}(\beta, \psi) - \tr\{\Sigma(\psi)^{-1}\}],
\end{align}
and, for $j < r$,
\begin{align} \label{eq:score_psi_j}
  \bar{S}(\psi_j) = \frac{1}{2}[\tilde{e}(\beta, \psi)^\tsp Z H_j Z^\tsp \tilde{e}(\beta, \psi) - \tr\{Z^\tsp \Sigma(\psi)^{-1}Z H_j\}].
\end{align}

For the first term in \eqref{eq:score_psi_j}, note $Z^\tsp \tilde{e}(\beta,
\psi)$ is shared by all $j < r$, so it is computed once and stored. Similarly,
the term $Z^\tsp \Sigma(\psi)^{-1}Z$ is shared by all $j < r$ and is, by
\eqref{eq:sigma_woodbury},
\begin{align*}
  \psi_r^{-1}\{Z^\tsp Z - Z^\tsp Z A(\psi) Z^\tsp Z\} = \psi_r^{-1}Z^\tsp Z  \{I_q- A(\psi) Z^\tsp Z\} = \psi_r^{-1} Z^\tsp Z F(\psi),
\end{align*}
where $F(\psi) = I_q - A(\psi) Z^\tsp Z$ is computed once and stored for future
computations. Note also that, once $Z^\tsp \Sigma(\psi)^{-1}Z$ is computed as in
the last display, it is cheap to compute $\tr\{Z^\tsp \Sigma(\psi)^{-1}Z H_j\} =
\mathtt{sum}[\{Z^\tsp \Sigma(\psi)^{-1}Z\} \circ H_j]$, where the $\circ$ is the
elementwise product and $\mathtt{sum}$ sums the elements of its argument; all
$r-1$ such traces cost $O(r q^2)$ in the worst case.

Having computed $F(\psi)$, $\bar{S}(\psi_r)$ adds only $O(n + q)$ operations, the $q$
coming from the trace and the $n$ from the inner product in
\eqref{eq:score_psi_r}. In particular, again by \eqref{eq:sigma_woodbury}, the
second term of \eqref{eq:score_psi_r} is $\tr\{\Sigma^{-1}(\psi)\} =
\psi_r^{-1}[n - \tr\{Z^\tsp Z A(\psi)\}] = \psi_r^{-1}[n - q + \tr\{F(\psi)\}]$.

We next focus on the expected information.  The expected information for $\psi$ has elements
\begin{align*}
  \bar{\mathcal{I}}_{ij}(\psi) =  \frac{1}{2} \times \begin{cases}
    \tr(Z^\tsp\Sigma^{-1}ZH_i Z^\tsp \Sigma^{-1}Z H_j ), & i \leq j < r \\
    \tr(Z^\tsp\Sigma^{-2} Z H_i ), & i < j = r \\
    \tr(\Sigma^{-2}), & i = j = r
  \end{cases},
\end{align*}
where we omitted the argument $\psi$ on the right-hand side for brevity.

For the first of the three cases, note we already computed $Z^\tsp \Sigma^{-1}Z$ for the score. Thus, we next compute $Z^\tsp \Sigma^{-1}Z H$, with $H = [H_1,
\dots, H_{r - 1}] \in \R{q \times q (r - 1)}$. Then, the traces needed for the
first case are the sums of the elementwise products of $q\times q$ blocks of
$Z^\tsp \Sigma^{-1}Z H$, which costs $O(r q^3)$ to form and $O(r^2 q^2)$ for the
pairwise block products without sparsity. For the second case, by \eqref{eq:sigma_woodbury},
\[
  Z^\tsp \Sigma^{-2}Z = \psi_r^{-1} Z^\tsp \Sigma^{-1}Z F(\psi).
\]
Hence, the required traces are the sums of elementwise products of $q\times q$
blocks of $Z^\tsp \Sigma^{-1}Z H$, which we already computed, and
$\psi_r^{-1}F(\psi)$. For the third case,
\[
  \tr(\Sigma^{-2}) = \psi_r^{-2}\{n - q + \mathtt{sum}(F \circ F^\tsp)\}.
\]
In summary, computing the log-likelihood value, score, and expected information
together has worst-case cost $O(nq + r q^3 + r^2 q^2)$.

\subsection{Restricted likelihood} \label{sec:reloglik}
Up to a constant, the restricted log-likelihood is
\begin{align} \label{eq:res_ll_comp}
  \ell(\psi) = \bar{\ell}(\tilde{\beta}(\psi), \psi) - \frac{1}{2}\log \vert X^\tsp \Sigma(\psi)^{-1}X\vert,
\end{align}
where $\tilde{\beta} = \{X^\tsp \Sigma(\psi)^{-1}X\}^{-1} X^\tsp
\Sigma(\psi)^{-1}Y$. We assume $X$ has full column rank $p < n$, so $X^\tsp
\Sigma^{-1}X$ is indeed invertible and $\tilde{\beta}(\psi)$ is the unique
partial maximizer of $\bar{\ell}(\beta, \psi)$ in $\beta$. The score for $\beta$ is
\begin{equation*}
  \bar{S}(\beta) = X^\tsp \Sigma(\psi)^{-1}(Y - X\beta) = X^\tsp \tilde{e}(\beta, \psi),
\end{equation*}
and, by \eqref{eq:sigma_woodbury}, the corresponding information block is
\begin{align} \label{eq:info_beta}
  \bar{\mathcal{I}}(\beta; \psi) = \cov_{\beta, \psi}\{\bar{S}(\beta; \psi)\} = X^\tsp \Sigma(\psi)^{-1}X = \psi_r^{-1}\{X^\tsp X - X^\tsp Z A(\psi) Z^\tsp X\}.
\end{align}

Differentiating \eqref{eq:res_ll_comp}, and using that the gradient of $\psi
\mapsto \bar{\ell}\{\tilde{\beta}(\psi), \psi\}$ is $\bar{S}(\psi;
\tilde{\beta}(\psi))$, the restricted score has elements
\begin{equation*}
  S(\psi_j) = \bar{S}(\psi_j; \tilde{\beta}(\psi)) - \frac{1}{2} \frac{\partial \log \vert X^\tsp \Sigma(\psi)^{-1}X\vert}{\partial \psi_j},
\end{equation*}
where
\[
  \frac{\partial \log \vert X^\tsp \Sigma^{-1}X\vert}{\partial \psi_j} = \begin{cases} -\tr\{Z^\tsp\Sigma^{-1} X (X^\tsp \Sigma^{-1}X)^{-1}X^\tsp \Sigma^{-1} Z H_j \} & j < r \\
   -\tr\{(X^\tsp \Sigma^{-1}X)^{-1}X^\tsp \Sigma^{-2} X \} & j = r \end{cases}.
\]

An alternative expression for the restricted score is useful. Let 
$$
P =  X(X^\tsp
\Sigma^{-1}X)^{-1}X^\tsp \Sigma^{-1}
$$
and $Q = I_n - P$. Note $PP = P$ and
$\Sigma^{-1} P = P^\tsp \Sigma^{-1}$, and hence $QQ = Q$ and $\Sigma Q^\tsp
\Sigma^{-1} = Q$. Since $QX = 0$, we get $Y - X\tilde{\beta} = Y - P Y = QY =
Q(Y - X\beta) = Q \Sigma^{1/2}\Sigma^{-1/2}(Y - X\beta) = Q \Sigma^{1/2} W$,
where $W = \Sigma^{-1/2}(Y - X\beta)$. Using this and that $W$ has a
multivariate standard normal distribution under $(\beta, \psi)$, one can show that, for
$j < r$,
\begin{equation} \label{eq:res_score_j}
  S(\psi_j) = \frac{1}{2} \{W^\tsp \Sigma^{1/2}Q^\tsp\Sigma^{-1}Z H_j
    Z^\tsp\Sigma^{-1} Q \Sigma^{1/2} W  -\tr(Q Z H_j
    Z^\tsp\Sigma^{-1} )\},
\end{equation}
and
\begin{equation} \label{eq:res_score_r}
  S(\psi_r) = \frac{1}{2} \{W^\tsp \Sigma^{1/2}Q^\tsp\Sigma^{-2} Q \Sigma^{1/2} W  -\tr( Q\Sigma^{-1})\}.
\end{equation}
The trace terms follow because the restricted score is the score based on the
likelihood for $V^\tsp Y$, so it has mean zero by the first Bartlett identity.
By the calculations above, the stochastic part of $S(\psi_j)$ is $(1/2)
W^\tsp \Sigma^{1/2}Q^\tsp\Sigma^{-1}Z H_j Z^\tsp\Sigma^{-1} Q \Sigma^{1/2} W$,
so the trace part must be
\[
  \tr(\Sigma^{1/2}Q^\tsp\Sigma^{-1}Z H_j Z^\tsp\Sigma^{-1} Q \Sigma^{1/2}) =
  \tr(Q Z H_j Z^\tsp \Sigma^{-1}),
\]
where we use $\Sigma Q^\tsp \Sigma^{-1} = Q$. The derivation for
$S(\psi_r)$ is similar.

For simplicity, we refer to the quadratic forms in \eqref{eq:res_score_j} and
\eqref{eq:res_score_r} as the stochastic parts of the restricted score, since
they are the only parts that depend on $Y$; the trace terms are
nonstochastic when $X$ is. The stochastic parts determine the information.
Specifically, using that $\cov_{\beta, \psi}(W^\tsp G_1 W, W^\tsp G_2 W) = 2 \tr(G_1G_2)$ for
any symmetric $A$ and $B$, we get
\begin{align} \label{eq:res_info}
  \mathcal{I}_{jk}(\psi) &= \frac{1}{2} \times \begin{cases}
       \tr(Z^\tsp \Sigma^{-1}QZH_j Z^\tsp \Sigma^{-1}Q Z H_k) & j\leq k < r\\
           \tr(Z^\tsp \Sigma^{-1}Q\Sigma^{-1}Q Z H_j) & j < k = r \\
    \tr(Q\Sigma^{-1} Q\Sigma^{-1}) & j = k = r
  \end{cases}.
\end{align}

We next discuss how to compute relevant quantities efficiently. First, using
\eqref{eq:info_beta}, it is cheap to compute $X^\tsp \Sigma^{-1}X$. A Cholesky
decomposition of that matrix is then used to compute both $\tilde{\beta}(\psi)$
and the determinant in \eqref{eq:res_ll_comp}. Then, to evaluate
$\ell(\psi)$, we need only compute $\bar{\ell}(\tilde{\beta}(\psi), \psi)$,
which is done as in Section \ref{sec:likelihood computation}, replacing $\beta$
by $\tilde{\beta}(\psi)$ everywhere. Similarly, we compute the stochastic parts
of $S(\psi)$ as the stochastic parts of $\bar{S}(\psi)$, replacing $\beta$ by
$\tilde{\beta}(\psi)$ everywhere.

To compute the nonstochastic parts of the restricted score and the information
matrix, we use \eqref{eq:res_score_j}--\eqref{eq:res_info}, as detailed in
Section~\ref{sec:res_compute}. The terms needed involve $X^\tsp
\Sigma^{-k}X$, $k \in \{1, 2, 3\}$; $Z^\tsp \Sigma^{-k} Z$, $k \in \{1, 2\}$;
and $X^\tsp \Sigma^{-k} Z$, $k \in \{1, 2\}$.

\subsection{Computing the restricted score and information} \label{sec:res_compute}

Let us introduce the notation
\begin{align*}
D_{(k)} &= (X^\tsp \Sigma^{-1}X)^{-1}X^\tsp \Sigma^{-k} X; \\
E_{(k)} &= (X^\tsp \Sigma^{-1}X)^{-1}X^\tsp \Sigma^{-k} Z.
\end{align*}
Note that $D_{(1)} = I_p$ and
\begin{align*}
  X^\tsp \Sigma^{-1} Z &= \psi_r^{-1}\{X^\tsp Z - X^\tsp Z A Z^\tsp Z\} = \psi_r^{-1}X^\tsp Z F(\psi),
\end{align*}
so that
\[
  E_{(1)} = \psi_r^{-1}(X^\tsp \Sigma^{-1}X)^{-1}X^\tsp Z F(\psi).
\]
For $k \geq 2$, we can compute $D_{(k)}$ and $E_{(k)}$ recursively, using
\begin{align*}
  D_{(k)} &= \psi_r^{-1}(X^\tsp \Sigma^{-1} X)^{-1}X^\tsp \Sigma^{-(k -1)}(I_n - Z A Z^\tsp)X \\
  &= \psi_r^{-1}\{D_{(k - 1)} - E_{(k - 1)} A Z^\tsp X\},
\end{align*}
and
\begin{align*}
  E_{(k)} &= \psi_r^{-1}(X^\tsp \Sigma^{-1}X)^{-1}X^\tsp \Sigma^{-(k - 1)}(I_n - Z A Z^\tsp)Z \\
  &= \psi_r^{-1}(E_{(k - 1)} - E_{(k - 1)} A Z^\tsp Z) \\
  &= \psi_r^{-1}E_{(k - 1)} F(\psi).
\end{align*}
 
 Expanding $Q = I_n - P$ in the trace terms of
\eqref{eq:res_score_j}--\eqref{eq:res_info}, and using $\Sigma Q^\tsp\Sigma^{-1}
= Q$ and $QQ = Q$, gives, with $M = Z^\tsp\Sigma^{-1}Z$ and $N =
Z^\tsp\Sigma^{-1}XE_{(1)}$,
\begin{align*}
  \tr(Q\Sigma^{-1}) &= \tr(\Sigma^{-1}) - \tr(D_{(2)}), \\
  \tr(QZH_jZ^\tsp\Sigma^{-1}) &= \tr(MH_j) - \tr(H_jN), \\
  \tr(Q\Sigma^{-1}Q\Sigma^{-1}) &= \tr(\Sigma^{-2}) - 2\tr(D_{(3)}) + \tr(D_{(2)}^2), \\
  \tr(Z^\tsp\Sigma^{-1}Q\Sigma^{-1}QZH_j) &= \tr(Z^\tsp\Sigma^{-2}ZH_j)
    - 2\tr(H_jZ^\tsp\Sigma^{-1}XE_{(2)}) \\
    &\qquad + \tr(H_jZ^\tsp\Sigma^{-1}XD_{(2)}E_{(1)}), \\
  \tr(Z^\tsp\Sigma^{-1}QZH_jZ^\tsp\Sigma^{-1}QZH_k) &= \tr\{(M - N)H_j(M - N)H_k\},
\end{align*}
where $\tr(D_{(3)}) = \psi_r^{-1}\{\tr(D_{(2)}) - \tr(E_{(2)}AZ^\tsp X)\}$ by
the recursion above, so that $D_{(3)}$ itself is never formed. The terms
$\tr(\Sigma^{-1})$, $\tr(\Sigma^{-2})$, $M$, and $Z^\tsp\Sigma^{-2}Z$ are
computed as in Section~\ref{sec:likelihood computation}.

\subsection{Inverting the test statistic}

Constructing a confidence region in practice requires evaluating $T(\psi_1)$ at many candidate values of $\psi_1$ and retaining those for which $T(\psi_1) \leq q_{1-\alpha}$. For clarity, we focus on a confidence interval for a single random effect variance, say $\psi_1$, with $\psi_{(2)}$ the nuisance parameter. Thus, we seek the $C(\alpha)$ of Section~\ref{sec:model},
\[
  \{\psi_1 \in [0, \infty) : T(\psi_1) \leq q_{1-\alpha}\}.
\]
Similar arguments apply to other parameters of interest, and the accompanying software implements the method for any parameter of interest.

In practice, confidence intervals are usually constructed after fitting the model, so (restricted) maximum likelihood estimates of $\beta$ and $\psi$ are available and provide natural starting values. If they are not available, any $(\beta^0, \psi^0) \in \R{p} \times \tilde{\Upsilon}$ such that $T(\psi_1^0) \leq q_{1-\alpha}$ can be used to initialize our algorithm.

The one-step option takes the Newton step \eqref{eq:onestep} from the value
accepted at the previous point of the search, and at the first point from the
nuisance components of the fitted model. In the simulations of
Section~\ref{sec:simulations}, intervals computed this way differ in width from
those obtained by profiling to convergence by at most $1.6\%$.

\end{document}